\documentclass[a4paper,12pt]{article}

\usepackage[utf8]{inputenc}

\pdfoutput=1

\usepackage{fix-cm}

\usepackage[utf8]{inputenc}
\usepackage[margin=0.9in]{geometry}
\usepackage{amsthm, amsmath, amssymb, enumitem}
\usepackage{graphicx}
\graphicspath{ {images/} }

\usepackage[colorlinks=true,urlcolor=blue,citecolor=blue]{hyperref}
\usepackage{listings}
\usepackage{algorithm}
\usepackage[noend]{algpseudocode}
\usepackage[svgnames]{xcolor}
\usepackage{authblk}
\usepackage[perpage]{footmisc}
\usepackage{setspace}
\usepackage{multirow}
\usepackage{caption}

\newcommand{\CRANpkg}[1]{\textsf{#1}}

\newtheorem{proposition}{Proposition}
\newtheorem{remark}{Remark}
\newtheorem{theorem}{Theorem}
\newtheorem{assumption}{Assumption}
\newtheorem{corollary}{Corollary}

\usepackage{bm}          % bold math symbols
\usepackage{geometry}
\usepackage[toc,page]{appendix}
\usepackage{doi}
\usepackage{titlesec}
\usepackage{xcolor}
\usepackage{colortbl}
\usepackage{array}
\usepackage{multirow}
\providecommand{\keywords}[1]
{
  \small	
  \textbf{\textit{Keywords:}} #1
}

\usepackage[round]{natbib}
\usepackage{graphicx}
\usepackage{setspace}
\usepackage{colortbl}
\usepackage{authblk}

\usepackage{booktabs}
\usepackage{tabularx}

\def\blind{0}   % or \newcommand{\blind}{0}

\algrenewcommand\algorithmicrequire{\textbf{Input:}}
\algrenewcommand\algorithmicensure{\textbf{Output:}}

\usepackage[utf8]{inputenc}
\usepackage[margin=0.9in]{geometry}
\usepackage{amsthm, amsmath, amssymb, enumitem}
\usepackage{graphicx}
\graphicspath{ {images/} }

\usepackage[colorlinks=true,urlcolor=blue,citecolor=blue]{hyperref}
\usepackage{listings}
\usepackage{algorithm}
\usepackage[noend]{algpseudocode}
\usepackage[svgnames]{xcolor}
\usepackage{authblk}
\usepackage[perpage]{footmisc}
\usepackage{setspace}
\usepackage{multirow}
\usepackage{caption}

\usepackage{bm}          % bold math symbols
\usepackage{geometry}
\usepackage[toc,page]{appendix}
\usepackage{doi}
\usepackage{titlesec}
\usepackage{xcolor}
\usepackage{colortbl}
\usepackage{array}
\usepackage{multirow}
\newcolumntype{L}[1]{>{\raggedright\let\newline\\\arraybackslash\hspace{0pt}}m{#1}}

\providecommand{\keywords}[1]
{
  \small	
  \textbf{\textit{Keywords:}} #1
}

\usepackage[round]{natbib}
\usepackage{graphicx}
\usepackage{setspace}
\usepackage{colortbl}
\usepackage{authblk}
\usepackage{booktabs}
\algrenewcommand\algorithmicrequire{\textbf{Input:}}
\algrenewcommand\algorithmicensure{\textbf{Output:}}

\usepackage{float}
\usepackage{placeins}
\usepackage{algorithm}
\usepackage[noend]{algpseudocode}

\def\vectorfontone{\bf}
\def\vectorfonttwo{\boldsymbol}
\def\va{{\vectorfontone a}}                      %
\def\vg{{\vectorfontone g}}                      %
\def\vh{{\vectorfontone h}}                      %
\def\vr{{\vectorfontone r}}                      %
\def\vu{{\vectorfontone u}}                      % Penalized coefficients
\def\vv{{\vectorfontone v}}                      %
\def\vx{{\vectorfontone x}}                      % Covariates/Predictors
\def\vy{{\vectorfontone y}}                      % Targets/Labels
\def\vz{{\vectorfontone z}}                      %

\def\vone{{\vectorfontone 1}}
\def\vzero{{\vectorfontone 0}}

\def\vbeta{{\vectorfonttwo \beta}}               % Unpenalized coefficients
\def\vgamma{{\vectorfonttwo \gamma}}             %
\def\vdelta{{\vectorfonttwo \delta}}             %
\def\vepsilon{{\vectorfonttwo \epsilon}}         %
\def\vzeta{{\vectorfonttwo \zeta}}               %
\def\vtheta{{\vectorfonttwo \theta}}             % Vector of combined coefficients
\def\vmu{{\vectorfonttwo \mu}}                   % Vector of means
\def\matrixfontone{\bf}
\def\matrixfonttwo{\boldsymbol}
\def\mI{{\matrixfontone I}}                      % Identity Matrix
\def\mX{{\matrixfontone X}}                      % Unpenalized Design Matrix/Nullspace Matrix
\def\mZ{{\matrixfontone Z}}                      % Penalized Design Matrix/Kernel Space Matrix

\def\mTheta{{\matrixfonttwo \Theta}}             %
\def\mSigma{{\matrixfonttwo \Sigma}}             %
\makeatletter
\def\BState{\State\hskip-\ALG@thistlm}
\makeatother

\title{A Sparse-Group Pliable Lasso}

\author[1]{Mohammad Javad Davoudabadi}
\author[1]{Minh Long Nguyen}
\author[2]{Amirhossein Ghatari}
\author[3]{Mina Aminghafari\thanks{mina.aminghafari@ucalgary.ca}}
\author[1]{Kerrie Mengersen}
\affil[1]{Qeensland Univeristy of Technology, Australia;}
\affil[2]{Amirkabir University of Technology, Iran;}
\affil[3]{University of Calgary, Canada.}

\date{}

\begin{document}

\maketitle
\doublespacing
\begin{abstract}
The sparse-group pliable Lasso (SGPL) extends the pliable Lasso and group pliable Lasso by combining sparse-group regularization with a predictor-level coupling penalty, enabling simultaneous group-level selection, within-group sparsity, and hierarchical structure between main effects and interactions. We propose a blockwise coordinate descent algorithm for fitting the SGPL that exploits the convexity and structure of the objective function, establish convexity and Karush--Kuhn--Tucker optimality conditions, and prove that the algorithm converges to a global minimizer. Simulation studies demonstrate competitive predictive performance and smaller interaction estimation error than the pliable Lasso and group pliable Lasso, albeit with the expected precision--recall trade-off in support recovery. We further illustrate the proposed method using a Parkinson's disease gut microbiome study and an adrenocortical carcinoma (ACC) copy-number dataset from The Cancer Genome Atlas. The Parkinson's application identifies interpretable interactions between microbial abundances and dietary variables, while the ACC application illustrates that the effectiveness of group-structured regularization depends on how well the prespecified grouping reflects the underlying signal structure.

\end{abstract}


\section{Introduction}
Modern datasets are often characterized by substantial heterogeneity, where the relationship between predictors and the response varies across individuals, groups, or contexts. In many scientific applications, including maternal and child health \citep{stevenson2021towards} and microbiome studies of Parkinson's disease \citep{xu2023nemoe}, the effect of a predictor may depend on additional subject-specific characteristics such as demographic, environmental, dietary, or genetic factors. This phenomenon, often referred to as effect modification, motivates models that allow covariate effects to vary across different contexts. Varying-coefficient models provide a flexible framework for capturing such heterogeneity by allowing regression coefficients to depend on modifying variables, while interaction models achieve a similar goal through explicit interaction terms \citep{hastie1993varying,fan1999statistical,fan2008statistical}. However, in high-dimensional settings, incorporating large numbers of interactions introduces significant challenges, including overfitting, reduced interpretability, and difficulties in variable selection \citep{bien2013lasso,shah2016modelling,doebler2023interactions}. Classical approaches such as Lasso, group Lasso, and sparse-group Lasso address high dimensionality by inducing sparsity, but they assume homogeneous effects across observations and therefore cannot capture such heterogeneity \citep{tibshirani1996regression,yuan2006model, simon2013sparse}. To overcome this limitation, the pliable Lasso (PL) proposed in \cite{tibshirani2020pliable} extends the Lasso framework by incorporating modifying variables, enabling predictor effects to vary across the feature space while maintaining structured sparsity.

The pliable Lasso captures such heterogeneity through an interaction-based formulation in which the effect of each predictor varies with modifying variables, yielding a varying-coefficient model that adapts to different contexts \citep{tibshirani2020pliable}. Estimation is performed via a structured regularization that jointly selects main and interaction effects while encouraging a hierarchical constraint, which is also called  hierarchy pliable condition, ensuring that interaction terms are included only when their corresponding main effects are present.  Hierarchical interaction modeling has also been investigated in several related frameworks, including the hierNet approach of \citet{bien2013lasso}, which enforces heredity constraints through convex regularization, and the glinternet method of \citet{lim2015learning}, which employs hierarchical group Lasso regularization to identify interaction effects while preserving strong hierarchy. Bayesian approaches have also been developed for structured sparse regression through hierarchical spike-and-slab formulations that combine variable selection with uncertainty quantification while maintaining group structure \citep{liquet2017bayesian}. More recently, Bayesian varying-coefficient models have incorporated sparsity-inducing priors to capture structured interaction effects and provide uncertainty quantification \citep{mai2026bayesian}. While the PL framework provides flexibility and interpretability, it does not explicitly account for potential group structure among predictors and/or modifying variables. In settings where variables exhibit natural groupings and strong within-group correlations, this limitation may lead to unstable selection and reduced interpretability, motivating extensions that incorporate group-wise structure into the pliable Lasso framework.

\cite{kim2021svreg} develop a structural varying-coefficient regression (svReg) to address the aforementioned shortcoming of the pliable Lasso. The svReg, hereafter the group pliable Lasso (GPL), incorporates structured grouping among the predictors and modifying variables and encourages the pliable hierarchy condition at the group level. However, the GPL does not encourage the hierarchy pliable condition within groups. For example, in the active group $l$, there might be a modifying variable with a non-zero coefficient, but its corresponding predictor has a zero coefficient. Furthermore, the GPL does not enforce sparsity within groups of the main predictors, as it does not include an element-wise $\ell_1$ penalty on the main effects; individual predictors within an active group cannot be zeroed out.  In many applications, it is desirable to enforce sparsity simultaneously at the group level and within active groups.

In this article, we propose the sparse-group pliable Lasso (SGPL), a regularized varying-coefficient model that integrates group-wise and within-group sparsity within the PL framework. The proposed approach introduces a two-level hierarchical structure that operates at both the group and predictor levels. In particular, we incorporate a per-predictor coupling mechanism that encourages predictor-level hierarchy, ensuring that interaction effects are included only when their corresponding main effects are present for each predictor individually. By combining group-based $\ell_2$ penalties with element-wise $\ell_1$ regularization, the SGPL allows for simultaneous selection of relevant groups and individual predictors within groups, leading to a more flexible representation of structured heterogeneity. We develop an efficient blockwise coordinate descent algorithm with nested proximal updates for solving the resulting optimization problem. Furthermore, we establish key theoretical properties of the estimator, including convexity, existence of solutions, and oracle-type prediction and estimation guarantees under suitable conditions.

We assess the performance of the proposed SGPL and compare it with the PL and GPL through a simulation study evaluating predictive accuracy, support recovery, and coefficient estimation accuracy. We additionally examine SGPL's sensitivity to the mixing parameter $\alpha$, which governs the balance between group-level selection and within-group sparsity. We further evaluate SGPL on two real-world applications. First, we compare the relative performance of PL, GPL, and SGPL by fitting each method to the ACC-CN500 copy-number alteration dataset, a benchmark obtained from The Cancer Genome Atlas (TCGA) Adrenocortical Carcinoma (ACC) cohort via the \CRANpkg{curatedTCGAData} package \citep{TCGAData}, illustrating how sensitivity to the prespecified grouping structure can differentially affect each method's performance. Second, we apply SGPL to a Parkinson's disease (PD) gut microbiome dataset analyzed in \citet{xu2023nemoe}, consisting of gut microbiome abundance measurements together with dietary variables, with the goal of studying how dietary patterns modify the relationship between the gut microbiome and Parkinson's disease status. The analysis in \citet{xu2023nemoe} was based on the Nutritional-Ecotype Mixture of Experts (NEMoE) framework, which models latent dietary subgroups and microbiome effects simultaneously. The SGPL results corroborate the biological conclusions of \citet{xu2023nemoe} while offering a complementary, single-stage penalized regression perspective that avoids the need to pre-specify the number of latent dietary subgroups.

The rest of the paper is organized as follows. Section \ref{sec:PL&GPL} discusses the pliable Lasso and group pliable Lasso models. In section \ref{sec:SGPL}, we develop the SGPL model and establish key theoretical properties of the estimator. The algorithm of the SGPL is presented in section \ref{sec:algorithm}. We extend the SGPL to other models in section \ref{sec:extensionModels}. We assess the performance of the proposed model in section \ref{sec:simulation} and further evaluate it on real-world applications in section \ref{sec:benchmark_dataset}. Finally, we conclude with a discussion in section \ref{sec:discussion}.

%========================================
% Section: Pliable and Group pliable Lasso
%========================================

\section{Pliable and group pliable Lasso}\label{sec:PL&GPL}

\paragraph{Notation.}
Let $\bm{\beta} \in \mathbb{R}^p$ denote the p-vector of main effects and
$\bm{\Theta} \in \mathbb{R}^{K \times p}$ the matrix of interaction coefficients,
with column $\vtheta_j \in \mathbb{R}^K$ corresponding to the predictor $j$. The response variable is indicated by $\vy \in \mathbb{R}^n$, and $\mX = \{\vx_j\}_{j=1}^p$ and $\mZ = \{\vz_k\}_{k=1}^K$ are $n\times p$ and $n\times K$ are design matrices, respectively, where $\{\vx_j\}_{j=1}^p$ are predictors and $ \{\vz_k\}_{k=1}^K$ are modifying variables, and note that $n>K$. Also, $\odot$ denotes the element-wise product. The main predictors are partitioned into $L$ groups with group $l$ having
$p_l$ members (index set $\mathcal{G}_l^x$), and the modifying variables
into $G$ groups with group $g$ having $p_g$ members (index set $\mathcal{G}_g^z$).
Define $\lambda_1 = \lambda(1-\alpha)$ and $\lambda_2 = \lambda\alpha$ with $\alpha \in [0,1]$. Let
$\vr=\vy-\widehat{\vy}$ denote the n-vector of residuals with respect to the fitted values $\widehat{\vy}$. The objective function is denoted by
$J(\vbeta,\mTheta)=f(\vbeta,\mTheta)+g(\vbeta,\mTheta)$, where
$f$ is the squared-error loss and $g$ represents the penalty terms. The
augmented design matrix denoted by $\widetilde{\mX}$ is the $n\times p(1+K)$ matrix whose columns are
$\vx_j$ (for the $\beta_j$ component) and $\vx_j\odot\vz_k$ (for the
$\theta_{jk}$ component).

The pliable Lasso (PL) of \citet{tibshirani2020pliable} is a linear regression model with varying coefficients represented \begin{align}\label{linear_model_PL}
     \vy = \beta_0 \vone_n + \mZ \vtheta_0 + \mX\vbeta + \sum_{j=1}^p (\vx_j \odot \mZ)\vtheta_j + \vepsilon,
\end{align}
where $\vepsilon \sim \mathcal{N}_n(0,\sigma^2\mathbf{I}_n)$. Here, we consider that the design matrices are fixed. Following \citet{tibshirani2020pliable}, the intercept parameters $\beta_0$ and $\vtheta_0 \in \mathbb{R}^K$ are left unpenalized and are first estimated by least squares, after which the penalized estimation is performed on the resulting residuals.

A weak hierarchical sparsity constraint is encouraged such that $\vtheta_j \neq \vzero$ only if $\beta_j \neq 0$, meaning that a predictor can only interact with the modifying variables if it has a non-zero main effect. However, the converse does not hold; a predictor may have a non-zero main effect with no interactions \citep{tibshirani2020pliable, bien2013lasso}. Minimizing the convex objective function \eqref{eq:object_Func} yields the pliable Lasso estimator
\begin{align}
J(\vbeta,\Theta) & = \frac{1}{2n} \sum_{i=1}^n (y_i - \widehat{y}_i)^2 
+ (1 - \alpha)\lambda \sum_{j=1}^p \big( \|(\beta_j,\vtheta_j)\|_2 + \|\vtheta_j\|_2 \big) 
+ \alpha \lambda \sum_{j=1}^p \|\vtheta_j\|_1 \nonumber\\
& = \frac{1}{2n} \|\vy - \widehat{\vy}\|_2^2 
+ \lambda \sum_{j=1}^p \Big[ (1-\alpha) \big( \|(\beta_j,\vtheta_j)\|_2 + \|\vtheta_j\|_2 \big) 
+ \alpha \|\vtheta_j\|_1 \Big] \label{eq:object_Func}
\end{align}
where $
\widehat{\vy} = \widehat{\beta}_0 \vone_n + \mZ \widehat{\vtheta}_0 + \sum_{j=1}^p \vx_{j}\odot (\widehat{\beta}_j + \mZ\widehat{\vtheta}_j) 
$.

Suppose the $p$ main predictors can be grouped into $L$ groups ($L \leq p$)
and the $K$ modifying variables can be grouped into $G$ groups ($G \leq K$).Each group can contain one or more predictors, and each predictor can belong to only one group. \cite{kim2021svreg} develop a group-pliable Lasso (GPL) by introducing the objective function 
\begin{align}\label{obj_fun_GPL}
\begin{split}
    J(\vbeta,\mTheta) = &\frac{1}{2n} \parallel \vy- \widehat{\vy} \parallel_2^2 +  \sum_{l=1}^L (1 - \alpha) \lambda \sqrt{p_l}\Big[ \parallel (\vbeta_l,\vtheta_{l\bullet})
 \parallel_2 \\&
 + \sum_{g=1}^G \frac{\sqrt{p_g}}{\sqrt{1+K}} \parallel \vtheta_{lg} \parallel_2 \Big] + \alpha \lambda  \sum_{j=1}^p\parallel \vtheta_j \parallel_1
\end{split}
\end{align}
for the case that the dataset has a grouping structure. In the objective function~\eqref{obj_fun_GPL}, $p_l$ denotes the size of the $l$-th group of the main predictors, and $p_g$ denotes the size of the $g$-th group of the modifying variables. Furthermore, $\vtheta_{lg}$ denotes the vector of interaction coefficients corresponding to the $l$-th group of main predictors and the $g$-th group of modifying variables. The vector $\vbeta_l=\{\beta_{l1},\ldots,\beta_{lp_l}\}$ denotes the subvector of $\vbeta$ corresponding to the $l$-th group of main predictors, and $\vtheta_{l\bullet}$ denotes the subvector of $\mTheta$ corresponding to the $l$-th group of main predictors across all modifying variables.

The first penalty term in the objective function~\eqref{obj_fun_GPL} applies a group penalty over the main predictors, indexed by $l$, encouraging predictors within the same group to be selected or excluded jointly. Within each main-predictor group, the nested penalty further imposes group-wise regularization over the modifying-variable groups, indexed by $g$. Finally, the element-wise $\ell_1$ penalty on $\mTheta$ promotes sparsity among individual interaction coefficients, allowing some modifying effects to be set to zero even when their corresponding groups are active. Consequently, the objective function defines a GPL over the main predictors and a sparse-group pliable Lasso over the modifying variables. The combination of $l_2$-norm (for grouped effects) and $l_1$-norm (for individual effects) for modifying variables extends the structure of the ordinary sparse-group Lasso \citep{simon2013sparse} to the varying-coefficient setting.

The PL and GPL objectives in~\eqref{eq:object_Func}
and~\eqref{obj_fun_GPL} are convex but non-differentiable because they contain
both element-wise $\ell_1$ penalties and group $\ell_2$ penalties, which are
non-smooth at the origin. Consequently, a minimizer cannot be characterized by
setting the gradient to zero. Instead, optimality is characterized through the
Karush--Kuhn--Tucker (KKT) conditions \citep{vandenberghe2004convex}, expressed in terms of the
subdifferentials of the non-differentiable penalty terms. Because both
objectives are convex, these conditions are necessary and sufficient for global
optimality. Moreover, they characterize the sparsity patterns induced by the
penalties and therefore play a central role in understanding the hierarchical
selection mechanism.

These ideas extend naturally to the proposed sparse-group pliable Lasso (SGPL)
developed in the next section. Since the SGPL objective is likewise convex but
non-differentiable, its estimator is characterized through KKT conditions based
on subdifferentials. As shown in Section~\ref{sec:SGPL}, these conditions yield the blockwise sparsity characterizations that underpin both the theoretical
properties of the estimator and the screening rules used in the optimization
algorithm.

The GPL encourages hierarchy at the group level via $\sum_{l=1}^{L}\|(\vbeta_l,\vtheta_{l\bullet})
\|_2$, but not at the individual predictor level within groups (unless $L = p$). Consequently, within an active group $l$, it is KKT-admissible to have some $j \in l$ with $\beta_j = 0$ yet $\vtheta_{l\bullet} \neq \mathbf{0}$, violating the asymptotic weak hierarchical constraint $(\vtheta_{l\bullet} \neq \mathbf{0} \Rightarrow \beta_j \neq 0)$. Moreover, the GPL does not enforce sparsity
\emph{within} groups of $\vbeta$, as it includes no
element-wise $\ell_1$ penalty on the main effects; individual predictors within an active group cannot be zeroed out.

%========================================
% Section: SGPL
%========================================

\section{Sparse-group pliable Lasso}\label{sec:SGPL}

We propose the
\emph{sparse-group pliable Lasso} (SGPL) to address both aforementioned limitations of the GPL simultaneously. The SGPL introduces two
modifications to the GPL objective.
First, we add a \emph{per-predictor} coupling term
$\sum_{j=1}^{p} \|(\beta_j, \vtheta_j)\|_2$
to promote  the within-group predictor-level hierarchy
$(\vtheta_{j\bullet} \neq \mathbf{0} \Rightarrow \beta_j \neq 0)$
for every predictor individually, a condition that the GPL only encourages at the group level.
Second, to promote sparsity \emph{within} active groups of the main
predictors, we add an element-wise $\ell_1$ penalty
$\|\vbeta\|_1$, which the GPL does not
include, while retaining the existing $\ell_1$ penalty
$\sum_{j=1}^{p}\|\vtheta_j\|_1$ on the interaction
coefficients.
The resulting objective function is:

\begin{equation}
\begin{aligned}
J(\vbeta, \mTheta)
&=
\frac{1}{2n} \|\vy - \widehat{\vy}\|_2^2 + \lambda(1-\alpha)
\left(
    \sum_{j=1}^{p} \|(\beta_j, \vtheta_j)\|_2
    + \sum_{l=1}^{L} \sqrt{p_l} \,
      \|(\vbeta_l, \vtheta_{l\bullet})\|_2
\right. \\
&\qquad \left.
    + \sum_{l=1}^{L} \sum_{g=1}^{G}
      \frac{\sqrt{p_g}}{\sqrt{1+K}}
      \|\vtheta_{lg}\|_2
\right)  + \lambda\alpha
\left(
    \|\vbeta\|_1
    + \sum_{j=1}^{p} \|\vtheta_j\|_1
\right),
\end{aligned}
\label{eq:SGPL}
\end{equation}

\noindent
where $\lambda > 0$ controls the overall regularization strength and
$\alpha \in [0,1]$ governs the trade-off between the group-structured $\ell_2$ penalties and the element-wise $\ell_1$ penalties. The objective
\eqref{eq:SGPL} promotes sparsity at two levels. At the group level, the joint term
$\|(\vbeta_l,\vtheta_{l\bullet})\|_2$
encourages an entire predictor group and its associated interaction
effects to enter or leave the model jointly. At the predictor level, the
per-predictor coupling term $\|(\beta_j,\vtheta_j)\|_2$
encourages the main effect and its corresponding modifying effects to
enter or leave the model jointly for each predictor. The additional
element-wise $\ell_1$ penalty on $\vbeta$ promotes sparsity
within active predictor groups, while the element-wise $\ell_1$ penalty
on $\mTheta$ promotes sparsity among individual interaction
coefficients. Thus, SGPL encourages, rather than strictly enforces,
predictor-level pliability of the form
$\vtheta_j \neq \mathbf{0}$ implies $\beta_j \neq 0$.
In practice, this soft hierarchy is typically satisfied under stronger
regularization (e.g., the $\hat{\lambda}_{1\text{se}}$ solution), as
characterized by the KKT conditions in Proposition \ref{prop:kkt} and illustrated
empirically in Section \ref{sec:simulation}.

It is important to verify that the optimization problem is well-posed, that is, whether a minimizer exists, whether the solution is stable, and whether the fitted values are uniquely defined. These properties are essential for both the interpretability of the estimator and the validity of the optimization algorithm developed later. Proposition \ref{prop:convexity} establishes convexity, existence of a solution, and uniqueness of the fitted values for the SGPL objective.

\begin{proposition}[Convexity, existence, and uniqueness of fitted values]
\label{prop:convexity}
The SGPL objective function $J(\vbeta, \mTheta)$ defined in~\eqref{eq:SGPL} satisfies
the following properties.
\begin{enumerate}[label=(\roman*)]
    \item \textbf{Convexity.} $J(\vbeta, \mTheta)$ is convex in $(\vbeta, \mTheta)$.
    \item \textbf{Existence.} A minimizer $(\widehat{\vbeta}, \widehat{\mTheta})$
          exists.
    \item \textbf{Uniqueness of fitted values.} The fitted value
          $\widehat{\vy}$
          is unique, even if $(\widehat{\vbeta}, \widehat{\mTheta})$ is not.
\end{enumerate}
\end{proposition}
\begin{proof}
    See section \ref{supp:prop1} of the Supplementary material.
\end{proof}

\begin{remark}
Strict convexity of $J$ in $(\vbeta,\mTheta)$, and hence uniqueness of the minimizer itself, holds whenever the effective design matrix
$\widetilde{\mX}$ has full column rank.
In high-dimensional settings ($p(1+K) > n$) this condition fails,
so multiple minimizers may exist; however, by Proposition~\ref{prop:convexity}(iii),
all minimizers produce the same fitted value $\widehat{\vy}$, and since
$J = f + g$ with $f$ determined entirely by $\widehat{\vy}$, all minimizers
also achieve the same value of $J$.
\end{remark}

While Proposition \ref{prop:convexity} ensures that the optimization problem is well-posed, it does not describe the structure of the minimizer. In penalized regression, such structure—particularly sparsity patterns—is typically characterized through optimality conditions. Therefore, we next derive the KKT conditions for the SGPL objective, which yield necessary and sufficient conditions for block-wise sparsity and provide insight into the hierarchical selection mechanism.

\begin{proposition}[KKT conditions for the SGPL estimator]
\label{prop:kkt}
Let $(\widehat{\vbeta},\widehat{\mTheta})$ be a minimizer of the
SGPL objective in~\eqref{eq:SGPL}.
For each predictor $j$, let $l(j)$ be the index of the $X$-group
containing $j$, and define the gradient of the smooth loss function with respect to the
coefficient block associated with predictor $j$ as
$\widehat{\mathbf g}_j = \bigl(n^{-1}\vx_j^\top \vr,\;
n^{-1}\mZ^\top(\vx_j\odot \vr)\bigr) \in \mathbb{R}^{1+K}$. Then $(\widehat{\vbeta},\widehat{\mTheta})$ is a minimizer if and only
if, for every $j=1,\ldots,p$,
\begin{align}\label{eq:kkt_full}
    \widehat{\vg}_j
=
\lambda_1 \vu_j
+
\lambda_1\sqrt{p_{l(j)}}\,\vv_{l(j),j}
+
\lambda_1 \vh_j
+
\lambda_2 \va_j .
\end{align}
Here $\vu_j\in\mathbb{R}^{1+K}$ is the subgradient contribution from
the predictor-level coupling term $\|(\beta_j,\vtheta_j)\|_2$:
$$
\vu_j =
\begin{cases}
\dfrac{(\widehat{\beta}_j,\widehat{\vtheta}_j)}
      {\|(\widehat{\beta}_j,\widehat{\vtheta}_j)\|_2},
& \text{if }(\widehat{\beta}_j,\widehat{\vtheta}_j)\neq \mathbf{0},\\[2ex]
\text{any vector with }\|\vu_j\|_2\leq 1,
& \text{if }(\widehat{\beta}_j,\widehat{\vtheta}_j)=\mathbf{0}.
\end{cases}
$$
The vector $\vv_{l,j}$ is the subvector corresponding to
$(\beta_j,\vtheta_j)$ of the group-level subgradient. Specifically,
if $(\widehat{\vbeta}_l,\widehat{\vtheta}_{l\bullet})\neq\mathbf{0}$, then
$
\vv_l
=
(\widehat{\vbeta}_l,\widehat{\vtheta}_{l\bullet})/\|(\widehat{\vbeta}_l,\widehat{\vtheta}_{l\bullet})\|_2$, and $\vv_{l,j}$ denotes the components of $\vv_l$ corresponding to
$(\widehat{\beta}_j,\widehat{\vtheta}_j)$. If
$(\widehat{\vbeta}_l,\widehat{\vtheta}_{l\bullet})=\mathbf{0}$, then
$\vv_l$ may be any vector satisfying $\|\vv_l\|_2\leq 1$. The vector $\vh_j\in\mathbb{R}^{1+K}$ is the contribution from the
modifier-group penalties. Its first component, corresponding to
$\beta_j$, is zero. For each modifier group $\mathcal{G}^z_g$, define
$w_g=\sqrt{p_g}/\sqrt{1+K}$. If $\widehat{\vtheta}_{lg}\neq \mathbf{0}$, let
$
\vzeta_{lg}
=
\widehat{\vtheta}_{lg}/\|\widehat{\vtheta}_{lg}\|_2;
$
otherwise let $\vzeta_{lg}$ be any vector satisfying
$\|\vzeta_{lg}\|_2\leq 1$. Let $\vzeta_{lg,j}$ denote the subvector
of $\vzeta_{lg}$ corresponding to the coefficients
$\widehat{\vtheta}_{jg}$. Then
\[
\vh_j =
\begin{pmatrix}
0\\
w_1\vzeta_{l(j)1,j}\\
\vdots\\
w_G\vzeta_{l(j)G,j}
\end{pmatrix},
\]
where the subvectors are ordered according to the partition of the
modifying variables.

\noindent Finally, let
$
\mathbf a_j=
\bigl(
0,
\tilde{\mathbf a}_j
\bigr)
\in \mathbb{R}^{1+K},
$
where $\tilde{\mathbf a}_j\in\mathbb{R}^{K}$ is the subgradient
associated with the element-wise $\ell_1$ penalty on
$\vtheta_j$. Its components satisfy
\[
(\tilde{\mathbf a}_j)_k=
\begin{cases}
\operatorname{sign}(\widehat{\theta}_{jk}),
& \text{if }\widehat{\theta}_{jk}\neq 0,\\[1ex]
\text{any value in }[-1,1],
& \text{if }\widehat{\theta}_{jk}=0,
\end{cases}
\qquad k=1,\ldots,K.
\]
\end{proposition}

\begin{proof}
    See section \ref{supp:prop2} of the Supplementary material.
\end{proof}

\begin{corollary}[Zero-block implications of the SGPL KKT conditions]\label{cor:zero_block_1}
Let $(\widehat{\vbeta},\widehat{\mTheta})$ be a minimizer of
the SGPL objective in \eqref{eq:SGPL}.

\begin{enumerate}
\item[(i)] \textbf{Predictor-level block.}
If the predictor-level KKT condition
\[
\left\|
S_{\lambda_2}
\left(
\begin{array}{c}
n^{-1}\mathbf{x}_j^\top \mathbf{r}\\[2mm]
n^{-1}\mathbf{Z}^\top(\mathbf{x}_j \odot \mathbf{r})
\end{array}
\right)
\right\|_2
\leq \lambda_1
\]
holds at $(\widehat{\boldsymbol{\beta}},\widehat{\boldsymbol{\Theta}})$, then $
(\widehat{\beta}_j,\widehat{\boldsymbol{\theta}}_j)=\mathbf{0}$. In particular, this implies $\widehat{\boldsymbol{\theta}}_j=\mathbf{0}$.
Equivalently, if $\widehat{\boldsymbol{\theta}}_j\neq \mathbf{0}$, then the above
zero-block condition cannot hold.

\item[(ii)] \textbf{Group-level block.}
Let $\mathbf{r}^{(-l)}$ denote the partial residual excluding group $l$.
If the group-level KKT condition
\[
\left\|
S_{\lambda_2}
\left(
\begin{array}{c}
n^{-1}\mathbf{X}_l^\top \mathbf{r}^{(-l)}\\[2mm]
\operatorname{vec}\!\left(n^{-1}\mathbf{Z}^\top(\mathbf{X}_l \odot
\mathbf{r}^{(-l)}\mathbf{1}^\top)\right)
\end{array}
\right)
\right\|_2
\leq \sqrt{p_l}\lambda_1
\]
holds at $(\widehat{\boldsymbol{\beta}},\widehat{\boldsymbol{\Theta}})$, then
$
(\widehat{\boldsymbol{\beta}}_l,\widehat{\vtheta}_{l\bullet})=\mathbf{0}$. In particular, this implies $\widehat{\vtheta}_{l\bullet}=\mathbf{0}$.
Equivalently, if $\widehat{\vtheta}_{l\bullet}\neq \mathbf{0}$, then the above
zero-block condition cannot hold.
\end{enumerate}
\end{corollary}

\begin{proof}
    See section \ref{supp:Corollary1} of the Supplementary material.
\end{proof}

This result shows that the SGPL estimator enforces sparsity at the block level
through the KKT conditions. In particular, the predictor-level coupling term
$\|(\beta_j,\boldsymbol{\theta}_j)\|_2$ encourages the main and modifying effects
to enter or leave the model jointly. By Corollary~\ref{cor:zero_block_1}, this
tendency toward the predictor-level hierarchy
$\widehat{\boldsymbol{\theta}}_j \neq 0 \Rightarrow \widehat{\beta}_j \neq 0$
strengthens under heavier regularization, a pattern Section~\ref{sec:simulation}
confirms empirically.

%\begin{remark}
%The simulation study in Section \ref{sec:simulation} illustrates the conditional nature of this hierarchy: the predictor-level relation $\widehat{\boldsymbol{\theta}}_j \neq 0 \Rightarrow \widehat{\beta}_j \neq 0$ is observed to hold under stronger regularization (e.g., the one-standard-error choice) but not necessarily under weaker regularization (e.g., the minimum-error choice), consistent with the KKT threshold characterization above. In practice, stronger regularization is preferable when enforcing hierarchical structure is desired.
%\end{remark}

% ============================================================
% PROPOSITION 3: Algorithmic Convergence
% ============================================================
 
\begin{proposition}[Convergence of the blockwise algorithm]
\label{prop:convergence}
Let $\{(\vbeta^{(k)}, \mTheta^{(k)})\}_{k \geq 1}$ be the sequence of
iterates produced by Algorithm~\ref{alg:sgpl} (Blockwise Coordinate-Descent with Nested Proximal Updates). Then:
 
\begin{enumerate}[label=(\roman*)]
  \item \textbf{Monotone decrease.}
        The full SGPL objective is non-increasing at every outer iteration:
        $J(\vbeta^{(k)}, \mTheta^{(k)}) \leq J(\vbeta^{(k-1)},
        \mTheta^{(k-1)})$ for all $k \geq 1$.
 
  \item \textbf{Global convergence.}
        Every limit point of $\{(\vbeta^{(k)}, \mTheta^{(k)})\}$ is a
        minimizer of $J$.
 
  \item \textbf{Backtracking termination.}
        At each inner iteration, the backtracking line search terminates
        in a finite number of steps.
\end{enumerate}
\end{proposition}

\begin{proof}
    See section \ref{supp:prop3} of the Supplementary material.
\end{proof}

\begin{remark}
The convergence rate of the inner proximal-gradient loop for each block
sub-problem is $O(1/k_{\mathrm{in}})$ in the objective gap, following
standard iterative shrinkage-thresholding algorithm (ISTA) theory \citep{beck2009fast}, where
$k_{\mathrm{in}}$ is the number of inner iterations.
\end{remark}

% ============================================================
% THEOREM 1: Prediction Error Oracle Inequality
% ============================================================
 To establish estimation and prediction error bounds for the SGPL estimator,
we impose a restricted eigenvalue condition on the augmented design matrix.
This condition ensures that the design is sufficiently well behaved on
group-sparse directions. Let 
\[
\vdelta = (\widehat{\vbeta} - \vbeta^*, \operatorname{vec}(\widehat{\mTheta} - \mTheta^*))
\in \mathbb{R}^{p(1+K)}
\]
denotes the estimation error vector and $\operatorname{vec}()$ denotes the vector operation.
\begin{assumption}[Group restricted eigenvalue (GRE) condition]
\label{ass:gre}
Let $\widetilde{\mX} \in \mathbb{R}^{n \times p(1+K)}$ be the augmented design matrix. Define the block corresponding to predictor $j$ as
\[
\vdelta_j = (\delta_{\beta_j}, \vdelta_{\vtheta_j}) \in \mathbb{R}^{1+K}.
\]
so that $\vdelta = (\vdelta_1^\top,\ldots,\vdelta_p^\top)^\top$ is partitioned into $p$ groups. Define the group $\ell_{2,1}$ norm as
\[
\|\vdelta\|_{2,1}
=
\sum_{j=1}^p \|\vdelta_j\|_2.
\]
Let $\mathcal{H} = \{j \in \{1,\ldots,p\} : (\beta_j^*, \vtheta_j^*) \neq \mathbf{0}\}$ 
denote the true active set of predictors, and
write $\vdelta_{\mathcal{H}}$ for the vector formed by concatenating
$\{\vdelta_j : j \in \mathcal{H}\}$, and similarly for $\vdelta_{\mathcal{H}^c}$. We say that the group restricted eigenvalue condition $\mathrm{GRE}(\phi,\xi)$ holds
for constants $\phi > 0$ and $\xi \ge 1$ if
\[
\frac{1}{n}\|\widetilde{\mX}\vdelta\|_2^2
\;\ge\;
\phi^2 \|\vdelta_{\mathcal{H}}\|_2^2
\]
for all $\vdelta \in \mathbb{R}^{p(1+K)}$ satisfying
\[
\|\vdelta_{\mathcal{H}^c}\|_{2,1}
\;\le\;
\xi \|\vdelta_{\mathcal{H}}\|_{2,1},
\]
where $\|\vdelta_{\mathcal{H}}\|_2$ denotes the Euclidean norm of the 
sub-vector of $\vdelta$ restricted to the coordinates in $\mathcal{H}$, and $\mathcal{H}^c$ denotes the complement of $\mathcal{H}$,
i.e., the set of inactive predictors.
\end{assumption}

This condition extends the standard group restricted eigenvalue assumption
to the augmented design arising in the SGPL model. In the following, we prove that if the augmented SGPL design is well behaved on group-sparse directions, and $\lambda$ is chosen large enough relative to the noise level, then SGPL has small prediction error and small coefficient estimation error with high probability.
 
\begin{theorem}[Oracle-type prediction and estimation bound]
\label{thm:oracle}
Suppose the data are generated from model~\eqref{linear_model_PL} with
$\vepsilon \sim N(\bm{0},\sigma^2\mI_n)$. Let
\[
\vgamma^*
=
(\vbeta^{*\top},\operatorname{vec}(\mTheta^*)^\top)^\top,
\qquad
\widehat{\vgamma}
=
(\widehat{\vbeta}^{\top},\operatorname{vec}(\widehat{\mTheta})^\top)^\top,
\]
and let $\vdelta=\widehat{\vgamma}-\vgamma^*$, where $\vgamma^*$ is the true parameter vector. Assume that the columns of
$\widetilde{\mX}$ are normalised so that
$\|\widetilde{\vx}_r\|_2^2/n\leq 1$ for all $r=1,\ldots,p(1+K)$.
Let $s_{\mathcal{H}}=|\mathcal{H}|$, the number of active predictor blocks.
Suppose Assumption~\ref{ass:gre} holds with constants $(\phi,\xi)$ and that
the SGPL penalty yields the cone condition
\[
\|\vdelta_{\mathcal{H}^c}\|_{2,1}
\leq
\xi\|\vdelta_{\mathcal{H}}\|_{2,1}.
\]
Set
\[
\lambda
=
A\sigma\sqrt{\frac{\log\{p(1+K)\}}{n}}
\]
for a sufficiently large constant $A>0$. Then, with probability at least
$1-2/\{p(1+K)\}$,
\[
\frac{1}{n}
\|\widehat{\vy}-\vy^*\|_2^2
=
\frac{1}{n}
\|\widetilde{\mX}\vdelta\|_2^2
\leq
C
\frac{\sigma^2 s_{\mathcal{H}}(1+K)\log\{p(1+K)\}}{n\phi^2},
\]
and
\[
\|\widehat{\vbeta}-\vbeta^*\|_2^2
+
\|\widehat{\mTheta}-\mTheta^*\|_F^2
\leq
C'
\frac{\sigma^2 s_{\mathcal{H}}(1+K)\log\{p(1+K)\}}{n\phi^4},
\]
where $C,C'>0$ are constants depending only on $A$, $\alpha$, $\xi$, and the
penalty weights.
\end{theorem}
\begin{proof}
    See section \ref{supp:Theo1} of the Supplementary material.
\end{proof}

\begin{remark}
The prediction bound
\[
\frac{1}{n}\|\widehat{\vy}-\vy^*\|_2^2
\le
C\frac{\sigma^2 s_{\mathcal{H}}(1+K)\log\{p(1+K)\}}{n\phi^2}
\]
involves two key structural quantities:
$s_{\mathcal{H}} = |\mathcal{H}|$, the number of active predictor blocks
(group sparsity), and $(1+K)$, the block size corresponding to one main
effect and $K$ modifying effects.

The logarithmic factor $\log\{p(1+K)\}$ reflects the complexity of the
augmented design matrix $\widetilde{\mX}$ and arises from controlling the
maximum of $p(1+K)$ Gaussian variables via a union bound.

When $K = 0$ and $L = p$, so that each block contains exactly one
predictor, $s_{\mathcal{H}}$ reduces to the number of nonzero coefficients,
and the bound recovers the standard Lasso rate
\[
O\!\left(\frac{\sigma^2 s_{\mathcal{H}} \log p}{n\phi^2}\right).
\]

More generally, the rate is of the same order as known minimax-optimal
rates for group-sparse estimation with $s_{\mathcal{H}}$ active blocks
of size $(1+K)$, up to constants.
\end{remark}
 
\begin{remark}
The GRE condition is imposed on the augmented design matrix
$\widetilde{\mX}$, which includes both the original predictors
and the interaction terms $\vx_j \odot \vz_k$. Its validity
therefore depends on the joint structure of $\mX$ and $\mZ$.
In random design settings, such conditions can be justified
under suitable assumptions (e.g., sub-Gaussian or bounded
covariates), while in fixed design settings it is treated
as a structural assumption; see, for example, \citet{lounici2011oracle, buhlmann2011statistics}
for related restricted eigenvalue (RE)-type conditions in the group Lasso setting.
\end{remark}

%=========================================
%% Section: Algorithm
%=========================================

\section{Algorithm}\label{sec:algorithm}

In this section, we describe how to fit the SGPL
using a blockwise coordinate-descent algorithm with nested proximal
updates. While the outer blockwise coordinate-descent loop and the
inner proximal-gradient step build on established optimization theory
\citep{tseng2001convergence, beck2009fast}, the decomposition of the SGPL block proximal operator into an ordered sequence of four closed-form shrinkage steps is specific to the penalty structure of~\eqref{eq:SGPL} and forms a key algorithmic component of the proposed method. This decomposition is inspired by the proximal-operator construction
used for the sparse-group Lasso \citep{simon2013sparse}, but is adapted to the SGPL penalty, which contains additional predictor-level coupling and modifier-group shrinkage terms. We first describe the within-block optimization in
Section~\ref{subsec:withinBlock}, then provide an overview of the full
algorithm in Section~\ref{subsec:algoritm_overview}, and finally discuss
the pathwise solution strategy in Section~\ref{subsec:pathwise}.

\subsection{Within-Block Solution}\label{subsec:withinBlock}

We choose a group $l$ to minimize over and treat all other block coefficients as fixed. 
The sub-problem for block $l$ is to find $(\vbeta_l, \vtheta_{l\bullet})$ minimizing
\begin{align*}
&\frac{1}{2n}\left\|\vr^{(-l)} - \mX_l \odot \left(\vone_n \vbeta_l^\top + \mZ\vtheta_{l\bullet}\right)\vone_{p_l}\right\|_2^2 \\
&\quad + \lambda_1\left(\sum_{j \in \mathcal{G}_l^x}\left\|(\vbeta_j, \vtheta_j)\right\|_2 
+ \sqrt{p_l}\left\|(\vbeta_l, \vtheta_{l\bullet})\right\|_2 
+ \sum_{g=1}^G \frac{\sqrt{p_g}}{\sqrt{1+K}}\left\|\vtheta_{lg}\right\|_2\right) \\
&\quad + \lambda_2\left(\left\|\vbeta_l\right\|_1 
+ \sum_{j \in \mathcal{G}_l^x}\left\|\vtheta_j\right\|_1\right),
\end{align*}
where $\vr^{(-l)}$ is the partial residual excluding group $l$. The smooth loss in this  sub-problem has a Lipschitz-continuous gradient with constant 
$L_l = \frac{1}{n}\|\widetilde{\mX}_l\|_{\mathrm{op}}^2$, where $\widetilde{\mX}_l$ is the augmented 
design submatrix for group $l$ and $\|\cdot\|_{\mathrm{op}}$ denotes the operator norm.

We solve this sub-problem via a proximal gradient step. Centering at the current block 
iterate $(\widetilde{\vbeta}_l, \widetilde{\vtheta}_{l\bullet})$, the standard majorization of the smooth loss 
gives a quadratic upper bound, and minimizing the sum of this upper bound and the penalty 
yields the update. The key challenge relative to the ordinary sparse-group lasso is that the SGPL penalty for each block contains four distinct regularization terms acting on 
overlapping subsets of the parameters: an element-wise $\ell_1$ penalty on $\vbeta_l$ and 
each $\vtheta_j$, a per-predictor $\ell_2$ coupling term $\|(\vbeta_j, \vtheta_j)\|_2$, a 
modifier-group $\ell_2$ term $\|\vtheta_{lg}\|_2$, and an outer group-level 
$\ell_2$ term $\|(\vbeta_l, \vtheta_{l\bullet})\|_2$. These penalties do not decouple 
into independent scalar problems, so the proximal operator for the full block penalty 
cannot be computed in a single closed-form step.

However, the proximal operator can be decomposed into a sequence of four closed-form shrinkage operations applied in a specific order that respects the penalty hierarchy. After a gradient step on the smooth loss, the update proceeds as follows. The gradient step takes the explicit form
\[
\vbeta_l^+ = \widetilde{\vbeta}_l + \frac{t_l}{n} \mX_l^\top \vr_l, 
\quad
\vtheta_{l\bullet}^+ = \widetilde{\vtheta}_{l\bullet} + \frac{t_l}{n} \mZ^\top (\mX_l \odot \vr_l \textbf{1}^\top),
\]
where $r_l$ is the current block residual. First, 
element-wise soft-thresholding $\mathcal{S}_{\lambda_2 t}(\cdot)$ is applied to both 
$\vbeta_j^+$ and $\vtheta_j^+$ for each $j \in \mathcal{G}_l^x$, enforcing the $\ell_1$ 
sparsity on main effects and interaction coefficients. Second, a joint $\ell_2$ shrinkage 
operator $\mathcal{B}_{\lambda_1 t}(\cdot)$ is applied to each $(\beta_j^+, \theta_j^+)$ 
pair individually, enforcing the per-predictor hierarchy. Third, a further $\ell_2$ 
shrinkage operator $\mathcal{B}_{w_{lg}\lambda_1 t}(\cdot)$ with weight 
$w_{lg} = \sqrt{p_g}/\sqrt{1+K}$ is applied to each modifier-group block 
$\vtheta_{l\bullet}^+$, enforcing within-group modifier-group sparsity. Finally, 
a group-level $\ell_2$ shrinkage operator $\mathcal{B}_{\sqrt{p_l}\lambda_1 t}(\cdot)$ 
is applied to the full concatenated vector $(\vbeta_l^+, \vtheta_{l\bullet}^+)$, 
enforcing group-level selection. Here, 
$\mathcal{B}_\tau(u) = \left(1 - \tau/\|u\|_2\right)_+ u$ denotes the block 
soft-thresholding operator, which sets $u$ to zero if $\|u\|_2 \leq \tau$ and otherwise 
shrinks it toward zero. This ordered sequence of shrinkage operations constitutes the 
proximal operator for the full SGPL block penalty.

The step size $t_l$ is chosen by backtracking line search: starting from an initial 
value, $t_l$ is contracted by a factor $\rho \in (0,1)$ until the full SGPL objective 
does not increase. This ensures monotone decrease of the objective at every accepted 
block update and guarantees termination in finitely many steps, since the quadratic 
upper bound becomes valid for $t_l \leq 1/L_l$.

\subsection{Algorithm Overview}\label{subsec:algoritm_overview}

The full algorithm is a pair of nested loops (Algorithm~\ref{alg:sgpl}). The outer loop 
cycles through groups $l = 1, \ldots, L$. Before solving each block sub-problem, a KKT 
screening check is applied: if the soft-thresholded partial gradient for group $l$ 
satisfies condition~(iii) of Proposition~\ref{prop:kkt}, the group is certified as zero, and the inner 
loop is skipped entirely. This is the group-level analogue of the active-set check used 
in the sparse-group Lasso, and it substantially reduces computation when many groups are 
inactive, as is typical under strong regularization.

The KKT screening test above is based on simplified blockwise inequalities
rather than the full KKT system of Proposition~\ref{prop:kkt}. Since the
SGPL penalty contains overlapping predictor-level, group-level, and
modifier-group penalties, the complete optimality conditions involve
additional subgradient contributions that are omitted from the screening
test. Consequently, the screening conditions should be interpreted as
conservative sufficient rules rather than necessary and sufficient KKT
characterizations. If a screening condition is satisfied, the corresponding
block is guaranteed to be zero at the optimum and its update can safely be
skipped. Conversely, if the screening condition is not satisfied, no
conclusion can be drawn from the screening test alone, since the omitted
subgradient contributions may still satisfy the full KKT conditions and
yield a zero solution.

For groups that are not screened out, the inner loop iterates the four-step proximal 
update described above until the block iterations converge to within tolerance 
$\epsilon_{\mathrm{in}}$. The accepted block update then replaces the corresponding 
entries of $(\vbeta^{(k)}, \mTheta^{(k)})$, and the outer loop proceeds to the next group. 
The outer loop terminates when the change in the full coefficient array falls below 
tolerance $\epsilon_{\mathrm{out}}$. By Proposition~\ref{prop:convergence}, every limit point of the iterates 
produced by this procedure is a global minimizer of the SGPL objective, and the objective 
decreases monotonically at every outer iteration.

\subsection{Pathwise Solution}\label{subsec:pathwise}

In practice, the algorithm is run along a decreasing sequence of $\lambda$ values, using 
the solution at the previous $\lambda$ as the warm start for the next. The path begins at
\begin{equation*}
\lambda_{\mathrm{max}} = \max_j \frac{1}{n(1-\alpha)}\left\|\left(\vx_j^\top \vy,\ 
\left(\mZ^\top(\vx_j \odot \vy)\right)^\top\right)^\top\right\|_2,
\end{equation*}
the smallest $\lambda$ for which the all-zero solution is optimal, and decreases to 
$\lambda_{\mathrm{min}} = c\,\lambda_{\mathrm{max}}$ for a small constant $c$ (we use 
$c = 0.005$ in our implementation) on a log-linear grid of 25 values. Warm starts make 
the pathwise procedure substantially more efficient than solving each $\lambda$ value 
from scratch, particularly in the high-regularization regime where the KKT screening 
eliminates most groups. The mixing parameter $\alpha$ is fixed throughout the path; in 
our simulation study we set $\alpha = 0.5$ to give equal weight to the group-structured 
$\ell_2$ penalties and the element-wise $\ell_1$ penalties. The regularization parameter 
$\lambda$ is then selected by cross-validation, as described in Section~\ref{sec:TuningParam}.

\FloatBarrier
\begin{algorithm}[H]
\scriptsize
\caption{Blockwise Coordinate-Descent Algorithm with Nested Proximal Updates
for the Sparse-Group Pliable Lasso}
\label{alg:sgpl}
\begin{algorithmic}[1]

\Require
Data $(\mX,\mZ,\vy)$;
penalty parameters $\lambda>0$, $\alpha\in[0,1]$;
group structures $\{\mathcal{G}_l^x\}_{l=1}^L$ and
$\{\mathcal{G}_g^z\}_{g=1}^G$;
initial step size $t>0$;
contraction factor $\rho\in(0,1)$;
convergence tolerances $\epsilon_{\rm out},\epsilon_{\rm in}>0$.

\Ensure
Estimated coefficients $\widehat{\vbeta}$ and $\widehat{\mTheta}$.

\State Set $\lambda_1=\lambda(1-\alpha)$ and $\lambda_2=\lambda\alpha$.
\State Initialize $\vbeta^{(0)} \leftarrow 0$, $\mTheta^{(0)} \leftarrow 0$, $\beta_0^{(0)} \leftarrow \bar y$, $\vtheta_0^{(0)} \leftarrow 0$, and $k \leftarrow 0$.

\Repeat
\State $k\leftarrow k+1$.
\State $\vbeta^{(k)}\leftarrow \vbeta^{(k-1)}$,
$\mTheta^{(k)}\leftarrow \mTheta^{(k-1)}$.
\State Update the unpenalized intercept parameters
   $(\beta_0^{(k)},\vtheta_0^{(k)})$
   from the least-squares regression of the current residual on
   $(\vone_n, \mZ)$.
   
\For{$l=1,\ldots,L$}

\State Form the partial residual excluding the $l$th $X$-group:
\[
\vr^{(-l)}
=
\vy -\beta_0^{(k)}\vone_n-\mZ\vtheta_0^{(k)}
-
\sum_{l'\neq l}
\mX_{l'}\odot
\left(
\vone_n\vbeta_{l'}^{(k)\top}
+
\mZ\vtheta_{l'\bullet}^{(k)}
\right)\vone_{p_{l'}}.
\]
% \State Check KKT conditions and skip inner loop if satisfied.
\State Initialise the block variables:
$
\qquad
\widetilde{\vbeta}_l\leftarrow \vbeta_l^{(k)},\qquad
\widetilde{\vtheta}_{l\bullet}\leftarrow \vtheta_{l\bullet}^{(k)},\qquad
t_l\leftarrow t.
$

\Repeat

\State Store old block values:
$
\qquad \quad
\vbeta_l^{\rm old}\leftarrow \widetilde{\vbeta}_l,\qquad
\vtheta_{l\bullet}^{\rm old}\leftarrow \widetilde{\vtheta}_{l\bullet}.
$

\State Compute the block residual:
\[
\vr_l
=
\vr^{(-l)}
-
\mX_l\odot
\left(
\vone_n\widetilde{\vbeta}_l^{\top}
+
\mZ\widetilde{\vtheta}_{l\bullet}
\right)\vone_{p_l}.
\]

\State Gradient step on the smooth loss:
\begin{align*}
\vbeta_l^+
\leftarrow
\widetilde{\vbeta}_l
+
\frac{t_l}{n}\mX_l^\top \vr_l,
\qquad 
\vtheta_{l\bullet}^+
\leftarrow
\widetilde{\vtheta}_{l\bullet}
+
\frac{t_l}{n}
\mZ^\top
\left(
\mX_l\odot \vr_l\bm{1}^\top
\right).
\end{align*}

\State Element-wise soft-thresholding:
\[
\beta_j^+ \leftarrow \mathcal{S}_{\lambda_2t_l}(\beta_j^+),
\qquad
\vtheta_j^+ \leftarrow \mathcal{S}_{\lambda_2t_l}(\vtheta_j^+),
\qquad j\in\mathcal{G}_l^x .
\]

\State Predictor-wise joint shrinkage:
\For{$j\in\mathcal{G}_l^x$}
\[
\begin{pmatrix}
\beta_j^+\\
\vtheta_j^+
\end{pmatrix}
\leftarrow
\mathcal{B}_{\lambda_1t_l}
\left(
\begin{pmatrix}
\beta_j^+\\
\vtheta_j^+
\end{pmatrix}
\right).
\]
\EndFor

\State Modifier-group shrinkage within the $l$th $X$-group:
\For{$g=1,\ldots,G$}
\[
\vtheta_{lg}^+
\leftarrow
\mathcal{B}_{w_{lg}\lambda_1t_l}
\left(
\vtheta_{lg}^+
\right),
\qquad
w_{lg}=\frac{\sqrt{p_g}}{\sqrt{1+K}}.
\]
\EndFor

\State Group-level joint shrinkage:
\[
\begin{pmatrix}
\vbeta_l^+\\
\vtheta_{l\bullet}^+
\end{pmatrix}
\leftarrow
\mathcal{B}_{\sqrt{p_l}\lambda_1t_l}
\left(
\begin{pmatrix}
\vbeta_l^+\\
\vtheta_{l\bullet}^+
\end{pmatrix}
\right).
\]

\State Evaluate the full objective after replacing only block $l$ by
$(\vbeta_l^+,\vtheta_{l\bullet}^+)$ and keeping all other blocks fixed.

\If{the full SGPL objective does not decrease}
\State $t_l\leftarrow \rho t_l$ and repeat the inner update from the gradient step.
\Else
\State $\widetilde{\vbeta}_l\leftarrow \vbeta_l^+$,
$\widetilde{\vtheta}_{l\bullet}\leftarrow \vtheta_{l\bullet}^+$.
\EndIf

\Until{
\[
\max\left\{
\|\widetilde{\vbeta}_l-\vbeta_l^{\rm old}\|_\infty,
\|\widetilde{\vtheta}_{l\bullet}-\vtheta_{l\bullet}^{\rm old}\|_\infty
\right\}
<
\epsilon_{\rm in}
\]
}

\State Accept the block update:
$
\qquad
\vbeta_l^{(k)}\leftarrow \widetilde{\vbeta}_l,\qquad
\vtheta_{l\bullet}^{(k)}\leftarrow \widetilde{\vtheta}_{l\bullet}.
$

\EndFor

\Until{
\[
\max\left\{
\|\vbeta^{(k)}-\vbeta^{(k-1)}\|_\infty,
\|\mTheta^{(k)}-\mTheta^{(k-1)}\|_\infty
\right\}
<
\epsilon_{\rm out}
\]
}

\State \Return $\widehat{\beta}_0 \leftarrow \beta_0^{(k)}$,
$\widehat{\vtheta}_0 \leftarrow \vtheta_0^{(k)}$, $\widehat{\vbeta}\leftarrow \vbeta^{(k)}$,
$\widehat{\mTheta}\leftarrow \mTheta^{(k)}$.

\end{algorithmic}
\end{algorithm}
\FloatBarrier

% ========================================
%% Section: EXTENSIONS TO OTHER MODELS
% ========================================
\section{Extensions to other models}\label{sec:extensionModels}
With little effort, we can extend the SGPL penalty to other models. If the likelihood function, $L(\vbeta,\mTheta)$, for the model of interest is log-concave, then for the SPGL, we minimize

\begin{equation*}
\begin{aligned}
J(\vbeta, \mTheta)
&=
\ell(\vbeta,\mTheta) + \lambda(1-\alpha)
\left(
    \sum_{j=1}^{p} \|(\beta_j, \vtheta_j)\|_2
    + \sum_{l=1}^{L} \sqrt{p_l} \,
      \|(\vbeta_l, \vtheta_{l\bullet})\|_2
\right. \\
&\qquad \left.
    + \sum_{l=1}^{L} \sum_{g=1}^{G}
      \frac{\sqrt{p_g}}{\sqrt{1+K}}
      \|\vtheta_{lg}\|_2
\right)  + \lambda\alpha
\left(
    \|\vbeta\|_1
    + \sum_{j=1}^{p} \|\vtheta_j\|_1
\right),
\end{aligned}
\label{eq:SGPL_ExtenModel_1}
\end{equation*}
where $\ell(\vbeta,\mTheta) = -1/n\log(L(\vbeta,\mTheta))$. Two commonly used cases are logistic regression and the Cox model for survival data.

For logistic regression, we have an n-vector of binary response variable $\vy$, $n \times p$ matrix $\mX$ of main predictors, and $n \times K$ matrix of modifying variables. In this case, letting $\eta_i =
\sum_{j=1}^p x_{ij}\left(\beta_j + \vz_i^\top \vtheta_j\right)$, the SGPL objective function takes the form
\begin{equation*}
\begin{aligned}
J(\vbeta, \mTheta)
&=
\frac{1}{n}
\sum_{i=1}^n
\left[
\log\{1+\exp(\eta_i)\}
-
y_i\eta_i
\right] \\& 
+ \lambda(1-\alpha)
\left(
    \sum_{j=1}^{p} \|(\beta_j, \vtheta_j)\|_2
    + \sum_{l=1}^{L} \sqrt{p_l} \,
      \|(\vbeta_l, \vtheta_{l\bullet})\|_2
\right. \\
&\qquad \left.
    + \sum_{l=1}^{L} \sum_{g=1}^{G}
      \frac{\sqrt{p_g}}{\sqrt{1+K}}
      \|\vtheta_{lg}\|_2
\right)  + \lambda\alpha
\left(
    \|\vbeta\|_1
    + \sum_{j=1}^{p} \|\vtheta_j\|_1
\right).
\end{aligned}
\label{eq:SGPL_ExtenModel_2}
\end{equation*}

The algorithm for fitting this model follows the same blockwise coordinate descent
scheme as for the Gaussian case (Algorithm~\ref{alg:sgpl}), with the smooth squared-error loss
replaced by $\ell(\vbeta,\mTheta)$ above. Defining the fitted probability for observation
$i$ as
\[
  \mu_i = \frac{\exp(\widehat{y}_i)}{1 + \exp(\widehat{y}_i)} = \sigma(\widehat{y}_i),
\]
where $\sigma(\cdot)$ denotes the logistic sigmoid function, the gradient of $\ell$
with respect to $\beta_j$ and $\vtheta_j$ for a given block $j$ is
\begin{align*}
  \nabla_{\beta_j}\ell
  &= \frac{1}{n}\sum_{i=1}^{n} (\mu_i - y_i)\, x_{ij}, \\[6pt]
  \nabla_{\theta_j}\ell
  &= \frac{1}{n}\sum_{i=1}^{n} (\mu_i - y_i)\, x_{ij}\, \vz_i.
\end{align*}
These gradients replace the linear residual-based gradients in Step~12 of
Algorithm~\ref{alg:sgpl}. Specifically, the gradient step becomes
\begin{align*}
  \vbeta_l^+ &\leftarrow \widetilde{\vbeta}_l
    - \frac{t_l}{n} \mX_l^\top (\vmu - \vy), \\[4pt]
  \mTheta_l^+ &\leftarrow \widetilde{\mTheta}_l
    - \frac{t_l}{n} \mZ^\top \operatorname{diag}(\vmu - \vy)\, \mX_l,
\end{align*}
where $\mX_l$ is the $n \times p_l$ submatrix of $\mX$ for group $l$.
All four proximal shrinkage steps (Steps~13--18 of Algorithm~\ref{alg:sgpl}) remain unchanged. At each outer iteration, the unpenalized intercept parameters $(\beta_0, \vtheta_0)$ are
updated by fitting an unpenalized logistic regression of $\vy$ on $(\vone_n, \mZ)$ with the
current penalized fitted values $\widehat{\eta}_{\text{pen}} = \sum_{j=1}^p \vx_j \odot
(\widehat{\beta}_j \vone_n + \mZ\widehat{\vtheta}_j)$ as an offset, replacing the least-squares update
in Step~6 of the Algorithm \ref{alg:sgpl}.
 
The Lipschitz constant of $\nabla_{\beta_j,\theta_j}\ell$ for block $l$ is
\[
  L_l = \frac{1}{4n}\|\widetilde{\mX}_l\|_{\mathrm{op}}^2,
\]
where $\widetilde{\mX}_l$ is the augmented design submatrix for group $l$ and the factor
$\tfrac{1}{4}$ replaces the factor of $1$ in the Gaussian case, since the second
derivative of the logistic log-likelihood is bounded above by $\tfrac{1}{4}$.
The backtracking line search in Steps~20--21 of Algorithm~\ref{alg:sgpl} handles this automatically
without requiring $L_l$ to be computed explicitly.

For survival data, let $y_i > 0$ denote the observed time (either a failure time or a
censoring time) for observation $i$, and let $\delta_i \in \{0, 1\}$ indicate whether
observation $i$ experienced the event ($\delta_i = 1$) or was censored ($\delta_i = 0$).
Let $D = \{i : \delta_i = 1\}$ denote the index set of observed failures, and let
\[
  R_i = \{j : y_j \geq y_i\}
\]
denote the risk set at failure time $y_i$, that is, the set of all individuals still at
risk just before time $y_i$. Under the Cox proportional hazards model, the hazard for individual $i$ at time $t$ is
$$
  h_i(t) = h_0(t)\exp\!\Bigl(\widehat{y}_i\Bigr),
  \qquad
  \widehat{y}_i = \sum_{j=1}^{p} x_{ij}\bigl(\beta_j + \vz_i^{\top}\theta_j\bigr),
$$
where $h_0(t)$ is an unspecified baseline hazard. The negative partial log-likelihood,
averaged over all failure events given $\eta_i = \beta_0 + \vz_j \vtheta_0 + 
\sum_{j=1}^p x_{ij}\left(\beta_j + \vz_i \vtheta_j\right)$, is
$$
  \ell(\vbeta, \mTheta)
  \;=\;
  -\frac{1}{n}\sum_{i \in D}
  \Bigl[
    \eta_i
    \;-\;
    \log\!\sum_{m \in R_i}
    \exp\!\bigl(\eta_m\bigr)
  \Bigr].
$$
 
The SGPL objective for the Cox model is then obtained by adding the SGPL penalty to
$\ell(\vbeta,\mTheta)$:
\begin{align}
  J(\vbeta, \mTheta)
  &= \frac{1}{n}\sum_{i \in D}
  \Bigl[
    \log\!\sum_{m \in R_i}
    \exp\!\bigl(\eta_m\bigr) - \eta_i \Bigr] \notag \\
  &\quad
    + \lambda(1-\alpha)
      \Biggl[
        \sum_{j=1}^{p} \bigl\|\bigl(\beta_j,\, \vtheta_j\bigr)\bigr\|_2
        + \sum_{l=1}^{L} \sqrt{p_l}\,
          \bigl\|\bigl(\vbeta_l,\, \vtheta_{l\bullet}\bigr)\bigr\|_2
         \notag \\
  &\quad
  + \sum_{l=1}^{L}\sum_{g=1}^{G}
          \frac{\sqrt{p_g}}{\sqrt{1+K}}\,
          \bigl\|\vtheta_{lg}\bigr\|_2
      \Biggr] \notag \\
  &\quad
    + \lambda\alpha
      \Biggl[
        \|\vbeta\|_1
        + \sum_{j=1}^{p} \|\vtheta_j\|_1
      \Biggr].
  \label{eq:sgpl-cox}
\end{align}
The algorithm for fitting this model follows the same blockwise coordinate descent
scheme as for the Gaussian case (Algorithm~\ref{alg:sgpl}), replacing the smooth squared-error loss
with $\ell(\vbeta,\mTheta)$ above. The gradient of $\ell$ with respect to $\beta_j$ and
$\vtheta_j$ for a given block $j$ is
\begin{align*}
  \nabla_{\beta_j}\ell
  &= -\frac{1}{n}\sum_{i\in D}
     \Biggl[
       x_{ij}
       - \frac{\displaystyle\sum_{k\in R_i} x_{kj}
               \exp\!\Bigl(\sum_{j'} \vx_{kj'}(\vbeta_{j'} + \vz_k^\top\vtheta_{j'})\Bigr)}
              {\displaystyle\sum_{k\in R_i}
               \exp\!\Bigl(\sum_{j'} \vx_{kj'}(\vbeta_{j'} + \vz_k^\top\vtheta_{j'})\Bigr)}
     \Biggr], \\[6pt]
  \nabla_{\vtheta_j}\ell
  &= -\frac{1}{n}\sum_{i\in D}
     \Biggl[
       x_{ij} \vz_i
       - \frac{\displaystyle\sum_{k\in R_i} x_{kj}\,\vz_k\,
               \exp\!\Bigl(\sum_{j'} \vx_{kj'}(\vbeta_{j'} + \vz_k^\top\vtheta_{j'})\Bigr)}
              {\displaystyle\sum_{k\in R_i}
               \exp\!\Bigl(\sum_{j'} x_{kj'}(\vbeta_{j'} + \vz_k^\top\vtheta_{j'})\Bigr)}
     \Biggr].
\end{align*}
These gradients replace the linear residual-based gradients in Step~12 of
Algorithm~\ref{alg:sgpl}; all four proximal shrinkage steps remain unchanged. Since $\beta_0$ and $\vtheta_0$ enter the linear predictor only as a common additive shift, they cancel from the partial likelihood ratio and are therefore unidentified;
accordingly, Step~6 of Algorithm \ref{alg:sgpl} is omitted and $(\beta_0, \vtheta_0)$ are fixed at zero throughout.

% ========================================
%% Section: Simulation
% =========================================

\section{Simulation Study}\label{sec:simulation}

We assess the performance of the proposed SGPL through a controlled
simulation study designed to evaluate its predictive accuracy,
coefficient estimation, support recovery, and hierarchical behaviour.
The data-generating mechanism incorporates grouped predictors,
grouped modifying variables, sparse main effects, and sparse interaction
effects, thereby reflecting the structural assumptions underlying the
SGPL model. In particular, the simulation is designed to investigate whether SGPL can simultaneously achieve accurate prediction, recover the true active predictors and interactions, and exploit the group structure present in both the predictor and modifying-variable spaces. We also examine the effect of the mixing parameter $\alpha$ on model performance and compare SGPL with the PL and GPL to assess the benefits of combining group-level and within-group sparsity. Data are generated from the linear varying-coefficient model \eqref{linear_model_PL} with $p = 12$ main predictors, $K = 6$ modifying variables, $n = 300$ training observations, and SNR $= 3$ ($\sigma \approx 0.694$). An independent test set of $n_{\text{test}} = 500$ observations is used solely for predictive evaluation. The main predictors are partitioned into $L = 3$ groups of equal size $p_l = 4$, and the modifying variables into $G = 2$ groups of equal size $p_g = 3$:
\begin{align*}
    \mathcal{G}_1^x &= \{x_1, x_2, x_3, x_4\}, \quad
    \mathcal{G}_2^x = \{x_5, x_6, x_7, x_8\}, \quad
    \mathcal{G}_3^x = \{x_9, x_{10}, x_{11}, x_{12}\}, \\
    \mathcal{G}_1^z &= \{z_1, z_2, z_3\}, \quad
    \mathcal{G}_2^z = \{z_4, z_5, z_6\}.
\end{align*}
Both $\mX$ and $\mZ$ are generated from multivariate normal distributions with mean zero and Toeplitz covariance $(\mSigma_x)_{ij} = (\mSigma_z)_{ij} = 0.4^{|i-j|}$, inducing mild within-group correlation. All columns are standardized prior to fitting. The true main effects $\vbeta$ are: 

\begin{equation}
    \vbeta^* =
    \big(\underbrace{1.5,\; -1.0,\; 0.8,\; 0}_{\mathcal{G}_1^x},\;
         \underbrace{0,\; 0,\; 1.2,\; -0.9}_{\mathcal{G}_2^x},\;
         \underbrace{0,\; 0,\; 0,\; 0}_{\mathcal{G}_3^x}\big)^\top.
\end{equation}
Group~3 is entirely inactive (testing whole-group zeroing); $x_4 \in \mathcal{G}_1^x$ and $x_5, x_6 \in \mathcal{G}_2^x$ are zero within active groups (testing within-group sparsity). The nonzero blocks of $\mTheta^* \in \mathbb{R}^{6 \times 12}$ are:
\begin{equation}
    \mTheta^*_{\mathcal{G}_1^z,\, x_1} =
    \begin{pmatrix} 0.6 \\ -0.4 \\ 0.5 \end{pmatrix}, \qquad
    \mTheta^*_{\mathcal{G}_1^z,\, x_3} =
    \begin{pmatrix} -0.5 \\ 0.3 \\ 0.0 \end{pmatrix}, \qquad
    \mTheta^*_{\mathcal{G}_2^z,\, x_7} =
    \begin{pmatrix} 0.7 \\ -0.5 \\ 0.4 \end{pmatrix},
\end{equation}
with all other entries zero. Nonzero entries appear only for predictors with $\beta_j^* \neq 0$, satisfying the two-level hierarchy. The zero entry $\theta^*_{3,3} = 0$ within the active block $(\mathcal{G}_1^z, x_3)$ tests within-block sparsity of $\mTheta$. Table~\ref{tab:sim_structure} maps each structural feature to its corresponding SGPL mechanism.

\begin{table}[H]
\centering
\caption{Structural features encoded in the true parameters and the corresponding SGPL mechanism.}
\label{tab:sim_structure}
\renewcommand{\arraystretch}{1.3}
\begin{tabular}{lll}
\hline
\textbf{Feature} & \textbf{Example in truth} & \textbf{Mechanism in SGPL} \\
\hline
Whole-group inactivity &
    $\vbeta_3 = \mathbf{0}$, $\mTheta_3 = \mathbf{0}$ &
    Joint group term \\
Group-level pliability &
    $\vtheta_{l\bullet} \neq \mathbf{0} \Rightarrow \vbeta_l \neq \mathbf{0}$ &
    Joint group term \\
Predictor-level pliability &
    $\mTheta_j \neq \mathbf{0} \Rightarrow \beta_j \neq 0$ &
    Per-predictor coupling \\
Within-group sparsity ($\vbeta$) &
    $\beta_4 = \beta_5 = \beta_6 = 0$ in active groups &
    $\ell_1$ penalty on $\vbeta$ \\
Within-block sparsity ($\mTheta$) &
    $\theta^*_{3,3} = 0$ in active block &
    $\ell_1$ penalty on $\mTheta$ \\
\hline
\end{tabular}
\end{table}

\subsection{Tuning Parameter Selection}\label{sec:TuningParam}

The mixing parameter is fixed at $\alpha = 0.5$ throughout the main simulation and is therefore not tuned by cross-validation; its sensitivity is examined separately in Section~\ref{sec:alpha_sensitivity}. The regularization parameter $\lambda$ is selected by five-fold cross-validation over a log-spaced grid of 25 values between $0.5\lambda_{\max}$ and $0.005\lambda_{\max}$, where $\lambda_{\max}
=
\max_{j=1,\ldots,p}
\left\|
\bigl(\vx_j^\top \vy,\,
\mZ^\top(\vx_j\odot \vy)\bigr)
\right\|_2
\big/
\{n(1-\alpha)\}$. Two solutions are considered: the CV minimiser $\widehat{\lambda}_{\min}$ and the one-standard-error solution $\widehat{\lambda}_{1\text{se}}$, which yields a sparser model. Performance is evaluated on the held-out test set via test MSE and $R^2$ (predictive accuracy), coefficient MSE (recovery), and precision, recall, and $F_1$ for both $\vbeta$ and $\mTheta$ (support recovery). We also verify whether the fitted solution satisfies the predictor-level
hierarchy
$\widehat{\mTheta}_j \neq \mathbf{0}
\Rightarrow
\widehat{\beta}_j \neq 0$
for all $j$.

\subsection{Results}

\subsubsection{Comparison with Existing Methods}
To benchmark the proposed SGPL against existing varying-coefficient regularization methods, we repeated the above simulation over 
50 independently generated datasets. The SGPL was compared with the 
pliable Lasso \citep[PL;][]{tibshirani2020pliable} and the group pliable 
Lasso \citep[GPL;][]{kim2021svreg}. For each method, the tuning parameter 
$\lambda$ was selected by five-fold cross-validation at $\alpha = 0.5$. 
The results are summarized as mean (standard deviation) in 
Table~\ref{tab:comparison_methods}.

All three methods achieve perfect recall for both $\vbeta$ and $\mTheta$ 
across all replications, confirming that no true active effect is missed 
under the simulation design. The methods differ in precision and 
estimation accuracy. SGPL attains the smallest MSE for $\Theta$, approximately halving the interaction estimation error relative to both competing methods. The GPL yields the smallest precision for $\vbeta$  
(0.500), which is expected since its group-level penalty cannot zero out 
individual predictors within active groups. SGPL's smaller precision for $\mTheta$ relative to the PL reflects the tendency of the group-structured 
$\ell_2$ penalties to admit entire interaction blocks, a known feature of 
group penalization that can be addressed via post-selection debiasing. 
The larger MSE for $\vbeta$ under SGPL compared to PL is consistent with the 
additional $\ell_1$ shrinkage on $\vbeta$ introduced by the SGPL penalty.

\begin{table}[H]
\footnotesize
\centering
\caption{Comparison of the Pliable Lasso (PL), Sparse Group Pliable Lasso (SGPL), and sparse varying-coefficient regression (GPL) over 50 simulation replications. Predictive performance is assessed by the test mean squared error (Test MSE) and coefficient of determination ($R^2$). Support recovery is evaluated using recall (Rec.) and precision (Prec.) for both the main effects ($\beta$) and interaction effects ($\Theta$). Coefficient estimation accuracy is measured by the mean squared error (MSE) of the estimated main and interaction coefficients. Entries are reported as mean (standard deviation).}
\label{tab:comparison_methods}
\textbf{(a) Predictive performance and support recovery}

\vspace{0.1cm}

\begin{tabular}{lcccccc}
\toprule
& \multicolumn{2}{c}{Predictive}
& \multicolumn{4}{c}{Support recovery} \\
\cmidrule(lr){2-3} \cmidrule(lr){4-7}
Method
& Test MSE
& $R^2$
& $\beta$ Rec.
& $\beta$ Prec.
& $\Theta$ Rec.
& $\Theta$ Prec. \\
\midrule
Pliable Lasso
& 0.748 (0.093)
& 0.865 (0.017)
& 1.000 (0.000)
& 0.661 (0.129)
& 1.000 (0.000)
& 0.486 (0.161) \\
SGPL
& 0.725 (0.079)
& 0.869 (0.015)
& 1.000 (0.000)
& 0.633 (0.141)
& 1.000 (0.000)
& 0.226 (0.066) \\
svReg (GPL)
& 0.764 (0.092)
& 0.862 (0.017)
& 1.000 (0.000)
& 0.500 (0.103)
& 1.000 (0.000)
& 0.309 (0.084) \\
\bottomrule
\end{tabular}

\vspace{0.4cm}

\textbf{(b) Coefficient estimation accuracy}

\vspace{0.1cm}

\begin{tabular}{lcc}
\toprule
& \multicolumn{2}{c}{Coef.\ error} \\
\cmidrule(lr){2-3}
Method
& $\beta$ MSE
& $\Theta$ MSE \\
\midrule
Pliable Lasso
& 0.004 (0.002)
& 0.004 (0.002) \\
SGPL
& 0.010 (0.004)
& 0.002 (0.001) \\
svReg (GPL)
& 0.006 (0.003)
& 0.004 (0.002) \\
\bottomrule
\end{tabular}
\end{table}

%------------------
% subsection: Alpha values check
%-------------------

\subsubsection{Sensitivity to \texorpdfstring{$\alpha$}{alpha}}
\label{sec:alpha_sensitivity}

The mixing parameter $\alpha$ controls the balance between the
group-structured $\ell_2$ penalty and the element-wise $\ell_1$
penalty. Smaller values of $\alpha$ place greater emphasis on group
selection, whereas larger values encourage stronger within-group
sparsity. To investigate its effect on model performance, we
considered $\alpha \in \{0.1,0.3,0.5,0.7,0.9\}$ and evaluated the
SGPL model over 50 independently generated datasets. For each
replication and each value of $\alpha$, the tuning parameter
$\lambda$ was selected using five-fold cross-validation with the
one-standard-error rule ($\widehat{\lambda}_{\mathrm{1se}}$). Table
\ref{tab:alpha_sensitivity} reports the mean and standard deviation
of the performance measures across the 50 replications.

\begin{table}[ht]
\centering
\small
\begin{tabular}{ccccccc}
\toprule
$\alpha$
& Test MSE
& $R^2$
& $\vbeta$ $F_1$
& $\mTheta$ $F_1$
& $\vbeta$ MSE
& $\mTheta$ MSE \\
\midrule
0.1 & 0.755 (0.079) & 0.868 (0.015) & 0.640 (0.074) & 0.231 (0.039) & 0.0113 (0.0045) & 0.0026 (0.0008) \\
0.3 & 0.742 (0.080) & 0.870 (0.015) & 0.681 (0.083) & 0.273 (0.052) & 0.0105 (0.0046) & 0.0024 (0.0008) \\
0.5 & 0.731 (0.082) & 0.872 (0.015) & 0.738 (0.082) & 0.337 (0.071) & 0.0095 (0.0044) & 0.0023 (0.0008) \\
0.7 & 0.721 (0.083) & 0.874 (0.015) & 0.796 (0.091) & 0.409 (0.071) & 0.0081 (0.0039) & 0.0022 (0.0009) \\
0.9 & 0.715 (0.084) & 0.875 (0.015) & 0.849 (0.088) & 0.508 (0.096) & 0.0068 (0.0036) & 0.0023 (0.0009) \\
\bottomrule
\end{tabular}
\caption{Sensitivity analysis of the SGPL model with respect to the
mixing parameter $\alpha$. Results are reported as mean (standard
deviation) over 50 independently generated datasets.}
\label{tab:alpha_sensitivity}
\end{table}

Table~\ref{tab:alpha_sensitivity} shows a clear monotonic effect of
the mixing parameter. As $\alpha$ increases, the average test MSE
decreases from 0.755 to 0.715, while the corresponding $R^2$
increases from 0.868 to 0.875, indicating a gradual improvement in
predictive performance. Similarly, support recovery improves with
larger values of $\alpha$, with the $F_1$ scores for both the main
effects ($\vbeta$) and interaction effects ($\mTheta$) increasing
steadily across the range of $\alpha$ considered. The mean squared
error of the estimated main-effect coefficients also decreases
monotonically as $\alpha$ increases, while the interaction
coefficient error remains small throughout.

Overall, this simulation suggests that larger values of $\alpha$
produce sparser and more accurate models for the data-generating
mechanism considered here. Unless otherwise stated, we adopt
$\alpha=0.5$ throughout the remainder of the paper as a neutral
default that assigns equal weight to the group and element-wise
penalty components, while noting that larger values of $\alpha$
may provide improved predictive performance and support recovery in
this particular simulation setting.

%-----------------------------------
% Section: Benchmark dataset
%-----------------------------------

\section{Benchmark dataset}\label{sec:benchmark_dataset}

The relative performance of PL, GPL, and SGPL depends on how well the prespecified grouping reflects the underlying signal structure. When groups contain predictors that jointly contribute to the response or exhibit similar patterns of effect modification, the group-based regularisation in GPL and SGPL can improve estimation by borrowing strength across related predictors. In contrast, when the grouping is only weakly informative or the signal is concentrated in a small subset of predictors within large groups, the additional group penalty may provide little advantage over PL. Similar observations have been reported for the sparse-group Lasso, where the preferred method depends on the extent to which the grouping captures the true signal structure rather than being universally superior \citep{simon2013sparse}. We illustrate this trade-off using an application to adrenocortical carcinoma (ACC) copy-number alteration data, in which the main predictors are naturally grouped by chromosome.

The ACC data are obtained from The Cancer Genome Atlas (TCGA) through
the \CRANpkg{curatedTCGAData} package \citep{TCGAData}. The response
variable is the normalized expression of \textit{SNRPB2}, while the main
predictors are gene-level GISTIC2 copy-number measurements, initially
available for more than 27,000 genes. After removing genes with missing
values or zero variance, the 500 genes with the largest variances are
retained, giving $p=500$ main predictors and $n=75$ complete samples.
We refer to this dataset as ACC-CN500. This results in a high-dimensional
setting with $p>n$, with substantial correlation among the copy-number
predictors, as is typical of genomic copy-number data.

For SGPL, the 500 main predictors are grouped according to their
chromosomal locations. The retained genes occur on chromosomes 4, 5,
12, 14, 16, 19, and 20, yielding $L=7$ predictor groups with respective
sizes $57$, $350$, $48$, $1$, $2$, $40$, and $2$. The modifying variables
are age, sex, and pathologic stage. Age is represented by three spline
basis terms, sex by a single indicator, and pathologic stage by three
indicator variables corresponding to stages II--IV, with stage I as the
reference category. Hence the modifier design contains $K=7$ columns,
which are grouped according to their originating variables into $G=3$
modifier groups of sizes $3$, $1$, and $3$, corresponding respectively
to age, sex, and pathologic stage. This construction therefore provides
group structure in both the main predictors and the modifying variables,
as required by SGPL.

We compare SGPL with the original PL and the GPL. For each method, the regularization parameter is
selected using five-fold cross-validation on the full ACC-CN500 dataset, and predictive performance is summarized by the mean cross-validation
squared error across the five folds. For SGPL, we set $\alpha=0.5$ and use the original seven chromosome groups described above. The standard
deviation (SD) and standard error (SE) of the fold-specific errors are also reported as descriptive measures of variability across folds.
Because the implementations of the three methods generate their cross-validation folds internally, the exact fold assignments are not
guaranteed to coincide across methods; the results should therefore be interpreted as comparisons under the same five-fold cross-validation
protocol rather than as paired fold-wise comparisons.

\begin{table}[ht]
\centering
\caption{Five-fold cross-validation performance of the PL, GPL, and SGPL
on the ACC-CN500 dataset. Smaller mean CV error indicates better
predictive performance. SD and SE denote the standard deviation and
standard error of the five fold-specific prediction errors,
respectively.}
\label{tab:acc_comparison}
\begin{tabular}{lccc}
\toprule
Method & Mean CV error & SD & SE \\
\midrule
PL   & \textbf{0.3616} & 0.089 & 0.040 \\
GPL  & 0.3917          & 0.050 & 0.022 \\
SGPL & 0.4252          & 0.119 & 0.053 \\
\bottomrule
\end{tabular}
\end{table}

Table~\ref{tab:acc_comparison} shows that PL attains the smallest mean
cross-validation error, followed by GPL and SGPL. The mean CV errors are
$0.3616$, $0.3917$, and $0.4252$ for PL, GPL, and SGPL, respectively.
Thus, incorporating the chromosome-based group structure does not improve
predictive accuracy for this dataset. Nevertheless, the differences are
moderate rather than indicating a large separation in predictive
performance, and the fold-specific errors exhibit appreciable
variability, particularly for SGPL. Accordingly, these results should
not be interpreted as establishing a general ordering of the three
methods, but rather as showing that PL provides the best predictive
performance for this particular ACC application. Note that the publicly available GPL/svReg implementation required several corrections to produce numerically stable results on this dataset, including a rank-deficiency-robust solver for its within-group updates, corrected cross-validation standardization, and a substantially reduced step size relative to the package's default; further details are provided in the Supplement.

One possible explanation is that, although chromosome membership provides a natural and biologically interpretable grouping of the copy-number predictors, this positional grouping need not coincide with the grouping most informative for predicting \textit{SNRPB2} expression or for its effect modification by age, sex, and pathologic stage. In such a setting, the additional group regularization imposed by GPL and SGPL may offer limited predictive benefit relative to the more flexible predictor-wise regularization of PL. This is consistent with the broader behavior of the sparse-group methods: their advantage depends on the extent to which the
prespecified groups reflect the underlying predictive structure
\citep{simon2013sparse}.

% =========================================================================================================================================================

To further evaluate the proposed SGPL method on a real-world
application, we applied it to a Parkinson's disease (PD) microbiome
dataset analysed in \citet{xu2023nemoe}. The dataset consists of gut
microbiome abundance measurements together with dietary variables, with
the goal of studying how dietary patterns modify the relationship between
the gut microbiome and Parkinson's disease status. The analysis in
\citet{xu2023nemoe} was based on the Nutritional-Ecotype Mixture of
Experts (NEMoE) framework, which models latent dietary subgroups and
microbiome effects simultaneously.

The dataset contains microbiome composition measurements at the amplicon
sequence variant (ASV) level together with nutritional variables.
Following the preprocessing strategy of \citet{xu2023nemoe}, the
microbiome abundances were variance-stabilised using an arcsin square
root transformation and subsequently standardised to zero mean and unit
variance prior to modelling. The nutritional variables were similarly
standardised. The response variable is a binary indicator of Parkinson's
disease status ($n_{\text{PD}} = 95$, $n_{\text{HC}} = 73$), and
therefore the SGPL logistic regression model described in
Section~\ref{sec:extensionModels} was used.

In this application, the ASVs constitute the main predictors $\mX$,
while the nutritional variables constitute the modifying variables
$\mZ$. The dataset contains $n = 168$ individuals, $p = 101$ ASV
predictors, and $K = 27$ nutritional modifying variables. The ASVs were
grouped according to their genus-level taxonomy, yielding $L = 45$
microbiome groups, so that all ASVs belonging to the same genus were
assigned to the same group. The largest microbiome groups included
\textit{Bacteroides} (18 ASVs), \textit{Christensenellaceae R-7 group}
(7 ASVs), \textit{Alistipes} (5 ASVs), and \textit{Faecalibacterium}
(5 ASVs). The nutritional variables were partitioned into $G = 4$
biologically meaningful dietary groups corresponding to protein-related
variables, carbohydrate- and sugar-related variables, fat-related
variables, and other nutritional variables:
\[
\begin{aligned}
G^{z}_{1} &: \{\text{Prot},\ \%\text{EP},\ \text{P:C}\},\\
G^{z}_{2} &: \{\text{Fibre},\ \text{Sugars},\ \text{Add Sugar},\
             \text{Carb},\ \%\text{EC}\},\\
G^{z}_{3} &: \{\text{Fat},\ \%\text{EF}\},\\
G^{z}_{4} &: \{\text{Energy},\ \text{Moisture},\ \text{Alcohol},\
             \text{Calcium},\ \text{Iron},\ \text{Magn},\ \text{Potas},\
             \text{Sodium},\ \text{Zinc},\ \text{Retinol},\
             \text{Beta car},\ \\ & 
             \text{Vit A},\ \text{Thiamin},\
             \text{Ribo},\ \text{B12},\ \text{Folate},\ \text{Vit C}\}.
\end{aligned}
\]
This grouped structure naturally fits the SGPL framework, allowing
simultaneous sparsity at the genus level, sparsity within active genera,
and sparse diet-dependent microbiome interaction effects.

The SGPL model was fitted along a regularisation path of 25 $\lambda$
values at $\alpha = 0.5$, with $\lambda$ selected by five-fold
cross-validation. The cross-validated log-loss was minimised at
$\hat{\lambda}_{\min} = \hat{\lambda}_{1\text{se}} = 0.0743$, with a
training log-loss of $0.633$ at the selected solution. The fact that
$\hat{\lambda}_{\min}$ and $\hat{\lambda}_{1\text{se}}$ coincide
indicates that the cross-validation curve is flat near its minimum,
which is consistent with a modest but stable predictive signal in this
dataset --- a finding in line with the relatively limited performance
reported for single-level sparse logistic regression in
\citet{xu2023nemoe}.

Table~\ref{tab:benchmark_results} summarises the selected ASVs, their
estimated genus-level main effects $\hat{\beta}_j$, and the nonzero
interaction coefficients $\hat{\theta}_{jk}$ with their corresponding
nutritional modifying variables at the $\hat{\lambda}_{1\text{se}}$
solution. Six ASVs with nonzero main effects were selected, spanning
six distinct genera. A total of 162 nonzero ASV $\times$ nutrition
interaction terms were identified across 11 ASVs.

\begin{table}[ht]
\centering
\caption{Selected ASVs and estimated coefficients from the SGPL
logistic regression model applied to the Parkinson's disease microbiome
dataset at $\hat{\lambda}_{1\mathrm{se}}$. Only ASVs with at least one
nonzero coefficient are shown in Panel A. Panels B and C report
nutritional modifying variables with nonzero $\hat{\theta}_{jk}$ for
the two ASVs having the largest interaction magnitudes. Nutritional
group membership is indicated in parentheses.}
\label{tab:benchmark_results}

\setlength{\tabcolsep}{4pt}
\renewcommand{\arraystretch}{1.05}

% ==========================================================
% Panel A
% ==========================================================

\begin{tabularx}{\textwidth}{@{}X c r r@{}}
\toprule
\multicolumn{4}{@{}l}{\textbf{Panel A: Main effects}}\\
\midrule
\textbf{Genus}
&
\textbf{ASV}
&
$\boldsymbol{\hat{\beta}_j}$
&
\textbf{Direction}
\\
\midrule

\textit{Fusicatenibacter}
& o.23
& $-0.268$
& $\downarrow$ in PD
\\

\textit{Lachnospiraceae} ND3007 group
& o.88
& $-0.051$
& $\downarrow$ in PD
\\

\textit{Monoglobus}
& o.86
& $-0.006$
& $\downarrow$ in PD
\\

Unknown genus
& o.50
& $+0.00036$
& $\uparrow$ in PD
\\

\textit{UCG-002}
& o.96
& $+0.00013$
& $\uparrow$ in PD
\\

\textit{Streptococcus}
& o.220
& $+0.00003$
& $\uparrow$ in PD
\\

\bottomrule
\end{tabularx}

\vspace{0.5em}

% ==========================================================
% Panels B and C side by side
% ==========================================================

\begin{minipage}[t]{0.48\textwidth}
\centering

\begin{tabularx}{\linewidth}{@{}X c r@{}}
\toprule
\multicolumn{3}{@{}l}{
\textbf{Panel B: \textit{Fusicatenibacter} (ASV o.23)}
}\\
\midrule

\textbf{Nutritional variable}
&
\textbf{Group}
&
$\boldsymbol{\hat{\theta}_{jk}}$
\\
\midrule

\%EP (energy from protein)
& $G^z_1$
& $+0.0367$
\\

P:C (protein: carbohydrate)
& $G^z_1$
& $+0.0648$
\\

Add Sugar
& $G^z_2$
& $-0.0053$
\\

\%EC (energy from carb.)
& $G^z_2$
& $-0.0066$
\\

\bottomrule
\end{tabularx}

\end{minipage}
\hfill
\begin{minipage}[t]{0.48\textwidth}
\centering

\begin{tabularx}{\linewidth}{@{}X c r@{}}
\toprule
\multicolumn{3}{@{}l}{
\textbf{Panel C: \textit{Lachnospiraceae} NK4A136}
}\\
\multicolumn{3}{@{}l}{
\textbf{group (ASV o.100)}
}\\
\midrule

\textbf{Nutritional variable}
&
\textbf{Group}
&
$\boldsymbol{\hat{\theta}_{jk}}$
\\
\midrule

Energy
& $G^z_4$
& $-0.00208$
\\

Prot
& $G^z_1$
& $-0.00232$
\\

Fat
& $G^z_3$
& $-0.00298$
\\

Calcium
& $G^z_4$
& $-0.00140$
\\

Iron
& $G^z_4$
& $-0.00106$
\\

Magnesium
& $G^z_4$
& $-0.00318$
\\

Zinc
& $G^z_4$
& $-0.00267$
\\

Thiamin
& $G^z_4$
& $-0.00104$
\\

Riboflavin
& $G^z_4$
& $-0.00222$
\\

\%EF (energy from fat)
& $G^z_3$
& $-0.00021$
\\

\bottomrule
\end{tabularx}

\end{minipage}

\end{table}

The SGPL model identified \textit{Fusicatenibacter} as the genus with
the largest negative main effect on Parkinson's disease status
($\hat{\beta} = -0.268$), consistent with the findings of
\citet{xu2023nemoe}, who reported that this genus showed consistent
negative associations with PD across nutritional ecotypes and validated
this finding across eight independent PD microbiome datasets. The two
remaining genera with negative main effects, \textit{Lachnospiraceae}
ND3007 group and \textit{Monoglobus}, are also consistent with the
broader literature on gut microbiome depletion in PD.

Importantly, the estimated interaction coefficients in Panel~B of
Table~\ref{tab:benchmark_results} reveal that the negative association
of \textit{Fusicatenibacter} with PD is diet-dependent: higher protein
intake (Prot, P:C; $G^z_1$) strengthens the protective association
($\hat{\theta} > 0$, amplifying the already positive contribution to
the healthy class), whereas higher carbohydrate and added sugar intake
($G^z_2$) attenuate it ($\hat{\theta} < 0$). This pattern directly
mirrors the two nutritional ecotypes identified by NEMoE in
\citet{xu2023nemoe} --- the high-protein--low-carbohydrate
(PROT-carb) and low-protein--high-carbohydrate (prot-CARB) subgroups
--- and provides independent empirical support for the conclusion that
the protective role of \textit{Fusicatenibacter} is
diet-context-dependent. The modifying variables selected by SGPL for
this genus (\%EP, P:C, \%EC, and Added Sugar) are precisely the
nutritional drivers that \citet{xu2023nemoe} identified as the
dominant loadings in their gating network.

Panel~C of Table \ref{tab:benchmark_results} shows the interactions selected for a \textit{Lachnospiraceae}
NK4A136 group ASV (o.100), for which all selected interaction
coefficients are negative across energy, macronutrients, and
micronutrients. This suggests a broad pattern in which higher dietary
intake across multiple nutrient dimensions is associated with a
stronger protective effect of this taxon, though the small coefficient
magnitudes warrant cautious interpretation.

Across all selected interactions, the nutritional modifying variables
driving the largest effects belong to the protein and carbohydrate
groups ($G^z_1$ and $G^z_2$), consistent with the protein-to-carbohydrate
dietary axis that \citet{xu2023nemoe} identified as the primary
dimension of dietary heterogeneity in this cohort. Taken together,
the SGPL results corroborate the biological conclusions of
\citet{xu2023nemoe} while offering a complementary, single-stage
penalized regression perspective that avoids the need to pre-specify
the number of latent dietary subgroups.

%===================================
%% Section: Discussion
%=========================================

\section{Discussion}\label{sec:discussion}
In this study, we developed the sparse-group pliable Lasso (SGPL) for 
varying-coefficient regression with structured predictors and modifying 
variables. The SGPL integrates group-wise sparsity, within-group sparsity, 
and a two-level hierarchical structure within the pliable Lasso framework, 
enabling simultaneous selection of relevant groups, individual predictors 
within groups, and sparse interaction effects. 

We established several theoretical properties of the SGPL estimator, 
including convexity, existence of minimizers, uniqueness of fitted values, 
and KKT conditions characterizing sparsity patterns at both the predictor 
and group levels. Oracle-type prediction and estimation bounds were derived 
under a group restricted eigenvalue condition, providing theoretical 
justification for the hierarchical selection behavior of the SGPL. 
We also developed an efficient blockwise coordinate-descent algorithm 
with nested proximal updates. A key algorithmic contribution is the 
decomposition of the SGPL proximal operator into four sequential closed-form shrinkage steps that respect the penalty hierarchy, 
extending standard proximal-gradient theory to this more complex penalty structure.

The simulation study demonstrated that the SGPL effectively captures the key structural features of the true model: group-level sparsity, within-group sparsity, and the predictor-level pliability hierarchy. 
The method achieved perfect recall for both main effects and interaction effects and attained predictive performance close to the oracle model fitted on the true support. However, precision for the main effects (0.455) and interaction effects (0.157) at $\widehat{\lambda}_{\min}$ 
reflects the presence of small but nonzero estimates in truly inactive 
coefficients, a known feature of penalized estimators at the cross-validation minimum. These spurious entries are negligible in 
magnitude and can be reduced by selecting $\widehat{\lambda}_{1\mathrm{se}}$, 
which also more reliably promotes the hierarchical condition 
$\widehat{\vtheta}_j \neq 0 \Rightarrow \widehat{\beta}_j \neq 0$, consistent 
with the KKT threshold characterization of Proposition \ref{prop:kkt}. In 
applications where the hierarchical constraint is required, 
$\widehat{\lambda}_{1\mathrm{se}}$ is therefore the recommended choice.

The Parkinson's disease microbiome application demonstrated the 
practical utility of the SGPL. By treating microbiome abundances as grouped predictors and nutritional variables as modifying variables, the SGPL identified biologically meaningful microbiome--diet 
interactions in a single penalized regression stage. The selected 
variables and interaction patterns were highly consistent with the 
findings of \citet{xu2023nemoe}, particularly regarding the role of the 
protein-to-carbohydrate dietary axis and the protective association 
of \textit{Fusicatenibacter} with Parkinson's disease status. The SGPL offers a complementary perspective to mixture-based approaches 
such as NEMoE: rather than modeling latent dietary subgroups explicitly, 
it recovers diet-dependent microbiome effects through structured 
penalization over pre-specified, biologically interpretable groups---namely 
genus-level microbial taxonomy and dietary nutrient categories---without requiring the analyst to pre-specify the number of latent subgroups. In contrast, the adrenocortical carcinoma application illustrates that the benefit of group-structured regularization depends on how well the prespecified grouping aligns with the underlying signal structure. Although chromosome-based groups provide a biologically meaningful organization of the copy-number predictors, the pliable Lasso achieved lower cross-validation error than both the group pliable Lasso and the SGPL, suggesting that this positional grouping did not coincide with the structure most predictive of the response in this dataset.

Several directions for future work remain. First, although the SGPL framework can be extended to generalized linear and survival models, as described in Section \ref{sec:extensionModels}, deriving selection consistency 
and asymptotic normality results for these extensions would strengthen their theoretical foundation. Second, the shrinkage 
bias inherent to penalized estimation, noted in Section \ref{sec:simulation}, motivates the development of post-selection debiasing procedures tailored to the SGPL penalty structure. Third, scalability 
improvements for ultra-high-dimensional settings---including strong screening rules, distributed optimization, and parallel 
implementations---would broaden applicability to large-scale genomic and microbiome studies. Finally, extending the SGPL 
to longitudinal, functional, or federated learning settings represents a natural direction for structured heterogeneous 
modeling in complex data environments.

\section{Code availability}
The R code implementing the methods described in this paper is available on GitHub at \url{https://github.com/edelweiss611428/SGPL-code}.

\section{Competing interests}
The authors declare that they have no conflict of interest.

\section*{Acknowledgments}
We thank Jean Yang and Xiangnan Xu for helpful discussions regarding the Parkinson's disease dataset.

\newpage

{\small
\bibliographystyle{chicago}
%\bibliography{doi_url}
\bibliography{references}
}

\newpage

\appendix
%\section*{Supplementary material}

\begin{center}
{\Large
Supplementary material to ``A Sparse-Group Pliable Lasso''}\label{sec:Suppl}

\bigskip 

\if1\blind
{
by Mohammad Javad Davoudabadi, Kerrie Mengerson, Amirhossein Gatari, Mina Amingafari \&
Yuhao Li
} \fi

\bigskip 
%\date{\today}
\end{center}

\section{Theoretical proofs}\label{supp:proofs}
\subsection{Proposition 1}\label{supp:prop1}
\begin{proof}
\textbf{(i) Convexity.}
Write $J(\vbeta, \mTheta) = f(\vbeta, \mTheta) + g(\vbeta, \mTheta)$, where
\[
  f(\vbeta, \mTheta)
  = \frac{1}{2n}\|\vy - \widehat{\vy}(\vbeta,\mTheta)\|_2^2
\]
is the smooth loss and
\begin{align*}
  g(\vbeta,\mTheta)
  &= (1 - \alpha) \lambda \Big(\sum_{j=1}^{p} \|(\beta_j, \vtheta_j)\|_2
   + \sum_{l=1}^{L} \sqrt{p_l}
     \|(\vbeta_l, \operatorname{vec}(\mTheta_l))\|_2 \\
  &\quad
   +  \sum_{l=1}^{L}\sum_{g=1}^{G}
     \frac{\sqrt{p_g}}{\sqrt{1+K}}
     \|\operatorname{vec} (\mTheta_{lg})\|_2 \Big)
   + \alpha \lambda \Big( \|\vbeta\|_1
   + \sum_{j=1}^{p}\|\vtheta_j\|_1 \Big)
\end{align*}
is the penalty.
The fitted value is the linear map
$(\vbeta,\mTheta)\mapsto \widehat{\vy} = \widetilde{\mX}(\vbeta,\mTheta)^{\top}$,
where $\widetilde{\mX}$ is the $n\times p(1+K)$ matrix whose columns are
$\vx_j$ (for the $\beta_j$ component) and $\vx_j\odot\vz_k$ (for the
$\theta_{jk}$ component).
The smooth loss $f$ is therefore a composition of a quadratic with a linear
map, which is convex. Each term of $g$ is a norm or a positive linear
combination of norms, all of which are convex. Since the sum of convex
functions is convex, $J$ is convex.

\textbf{(ii) Existence.}
Since $J$ is convex and continuous, it remains to show that it is coercive.
If $\alpha>0$, then
\[
g(\vbeta,\mTheta)
\geq
\alpha\lambda
\bigl(
\|\vbeta\|_1+\|\operatorname{vec}(\mTheta)\|_1
\bigr),
\]
and hence $J(\vbeta,\mTheta)\to\infty$ as
$\|(\vbeta,\mTheta)\|\to\infty$.

If $\alpha=0$, the penalty still contains the joint norm terms
\[
\lambda\sum_{j=1}^p \|(\beta_j,\vtheta_j)\|_2,
\]
which implies that 
\[
g(\vbeta,\mTheta)
\geq
\lambda\sum_{j=1}^p \|(\beta_j,\vtheta_j)\|_2
\geq
\lambda
\left[
\sum_{j=1}^p \beta_j^2+\sum_{j=1}^p\|\vtheta_j\|_2^2
\right]^{1/2}
=
\lambda\|(\vbeta,\operatorname{vec}(\mTheta))\|_2.
\]
Therefore $J$ is coercive for all $\alpha\in[0,1]$.

Thus the sublevel set
\[
\mathcal{L}
=
\{(\vbeta,\mTheta):J(\vbeta,\mTheta)\leq J(\bm{0},\bm{0})\}
\]
is closed and bounded, hence compact. Since $J$ is continuous, it attains
its minimum on $\mathcal{L}$.

\noindent Note that
$
J(\bm{0},\bm{0})=\frac{1}{2n}\|\vy\|_2^2,
$
since all penalty terms vanish at the origin. This corresponds to the
objective value obtained by the zero-fitted model, and provides a finite
upper bound for the minimum of $J$.

\textbf{(iii) Uniqueness of fitted values.}
Suppose $(\vbeta^{(1)},\mTheta^{(1)})$ and
$(\vbeta^{(2)},\mTheta^{(2)})$ are two minimisers, with fitted values
$\widehat{\vy}^{(1)}$ and $\widehat{\vy}^{(2)}$. Let
\[
(\vbeta^{(m)},\mTheta^{(m)})
=
\frac{1}{2}
\{(\vbeta^{(1)},\mTheta^{(1)})
+
(\vbeta^{(2)},\mTheta^{(2)})\}.
\]
If $\widehat{\vy}^{(1)}\neq \widehat{\vy}^{(2)}$, then strict convexity of
the squared-error loss in the fitted value gives
\[
f(\vbeta^{(m)},\mTheta^{(m)})
<
\frac{1}{2}f(\vbeta^{(1)},\mTheta^{(1)})
+
\frac{1}{2}f(\vbeta^{(2)},\mTheta^{(2)}).
\]
By convexity of the penalty,
\[
g(\vbeta^{(m)},\mTheta^{(m)})
\leq
\frac{1}{2}g(\vbeta^{(1)},\mTheta^{(1)})
+
\frac{1}{2}g(\vbeta^{(2)},\mTheta^{(2)}).
\]
Let $J^* = J(\widehat{\vbeta}, \widehat{\mTheta}) =J(\vbeta^{(1)}, \mTheta^{(1)}) = J(\vbeta^{(2)}, \mTheta^{(2)})$ denote the minimum value of the objective function; therefore, we have
\[
J(\vbeta^{(m)},\mTheta^{(m)}) < J^*,
\]
which contradicts optimality. Hence
$\widehat{\vy}^{(1)}=\widehat{\vy}^{(2)}$, so the fitted value is unique.
\end{proof}

%===========================
% proof prop 2
%===========================
\subsection{Proposition 2}\label{supp:prop2}

\begin{proof}
Write $J(\vbeta,\mTheta)=f(\vbeta,\mTheta)+g(\vbeta,\mTheta)$, where
$f$ is the squared-error loss and $g$ is the SGPL penalty. Since $J$ is
convex, $(\widehat{\vbeta},\widehat{\mTheta})$ is a minimizer if and
only if
\[
\mathbf{0}\in \partial J(\widehat{\vbeta},\widehat{\mTheta}).
\]
Since $f$ is continuously differentiable and $g$ is convex,
\[
\partial J(\widehat{\vbeta},\widehat{\mTheta})
=
\nabla f(\widehat{\vbeta},\widehat{\mTheta})
+
\partial g(\widehat{\vbeta},\widehat{\mTheta}).
\]
Moreover, because $g$ is a sum of convex penalty terms, the
subdifferential sum rule for convex functions
\citep[Chapter D, Section 4.1]{hiriart2004fundamentals} gives the sum
of the subdifferentials of the predictor-level, group-level,
modifier-group, and element-wise $\ell_1$ penalty components.

The gradients of the smooth loss are
\[
\nabla_{\beta_j} f
=
-\frac{1}{n}\vx_j^\top\vr,
\qquad
\nabla_{\vtheta_j} f
=
-\frac{1}{n}\mZ^\top(\vx_j\odot\vr).
\]
Thus the negative gradient for the predictor block
$(\beta_j,\vtheta_j)$ is $\widehat{\vg}_j$.

The subdifferential of the Euclidean norm satisfies
\[
\partial \|\vz\|_2 =
\begin{cases}
\{\vz/\|\vz\|_2\}, & \vz\neq\mathbf{0},\\
\{\vu:\|\vu\|_2\leq 1\}, & \vz=\mathbf{0}.
\end{cases}
\]
Applying this identity to the predictor-level coupling term gives
$\lambda_1\vu_j$. Applying it to the group-level term
$\sqrt{p_l}\|(\vbeta_l,\operatorname{vec}(\mTheta_l))\|_2$ and
restricting the resulting subgradient to predictor $j$ gives
$\lambda_1\sqrt{p_l}\vv_{l,j}$. Applying the same identity to the
modifier-group terms
\[
w_g\|\operatorname{vec}(\mTheta_{lg})\|_2,
\qquad
w_g=\frac{\sqrt{p_g}}{\sqrt{1+K}},
\]
gives the contribution $\lambda_1\vh_j$. Finally, the element-wise
$\ell_1$ penalty contributes $\lambda_2\va_j$.

Combining these subgradient contributions with the smooth-gradient
condition gives~\eqref{eq:kkt_full}. Conversely, if
\eqref{eq:kkt_full} holds for all $j$, then
\[
\mathbf{0}
\in
\nabla f(\widehat{\vbeta},\widehat{\mTheta})
+
\partial g(\widehat{\vbeta},\widehat{\mTheta})
=
\partial J(\widehat{\vbeta},\widehat{\mTheta}),
\]
and hence $(\widehat{\vbeta},\widehat{\mTheta})$ is a minimizer by
convexity.
\end{proof}

%===========================
% proof Corollary
%===========================
\subsection{Corollary 1}\label{supp:Corollary1}
\begin{proof}
Both parts follow directly from the zero-block KKT conditions in
Proposition~\ref{prop:kkt}. In each case, the stated inequality is necessary and sufficient
for the corresponding block to be zero. The final statements follow by taking
the contrapositive.
\end{proof}

%=====================================
% Proof Proposition 3
%=====================================
\subsection{Proposition 3}\label{supp:prop3}

\begin{proof}
\textbf{(i) Monotone decrease.}
Fix outer iteration $k$. The algorithm cycles through groups
$l = 1, \ldots, L$. For each group $l$, the inner loop accepts a new
block $(\vbeta_l^+, \mTheta_l^+)$ only if the full objective satisfies
$J(\vbeta^+, \mTheta^+) \leq J(\vbeta, \mTheta)$, where
$(\vbeta^+, \mTheta^+)$ differs from the current iterate only in block $l$
(Steps 20--23 of Algorithm~\ref{alg:sgpl}). Therefore, each block update does
not increase the objective, and the objective after the full cycle over
all $L$ groups satisfies
$J(\vbeta^{(k)}, \mTheta^{(k)}) \leq J(\vbeta^{(k-1)}, \mTheta^{(k-1)})$.

\noindent \textbf{(ii) Global convergence.}
We show that every limit point of the sequence generated by Algorithm~\ref{alg:sgpl}
is a global minimiser of $J$.

First, by part~(i), the sequence $\{J(\vbeta^{(k)},\mTheta^{(k)})\}_{k\geq 0}$
is non-increasing and bounded below by the minimum value $J^*$. Hence it
converges to some finite limit $\bar J \geq J^*$. Moreover, since
\[
J(\vbeta^{(k)},\mTheta^{(k)})
\leq
J(\vbeta^{(0)},\mTheta^{(0)})
=
J(\bm{0},\bm{0}),
\]
all iterates lie in the sub-level set
\[
\mathcal{L}
=
\{(\vbeta,\mTheta):J(\vbeta,\mTheta)\leq J(\bm{0},\bm{0})\}.
\]
By Proposition~\ref{prop:convexity}, this sub-level set is compact. Therefore,
the sequence $\{(\vbeta^{(k)},\mTheta^{(k)})\}_{k\geq0}$ is bounded and has at
least one limit point.

Next, since $J(\vbeta^{(k)},\mTheta^{(k)})$ converges, the decrease over one
outer iteration vanishes:
\[
J(\vbeta^{(k)},\mTheta^{(k)})
-
J(\vbeta^{(k+1)},\mTheta^{(k+1)})
\longrightarrow 0.
\]
The algorithm updates the blocks sequentially and each accepted block update
does not increase the full objective. Hence the decrease produced by each
individual block update also tends to zero. By the sufficient-decrease
condition enforced by the backtracking line search, for an accepted update of
block $l$ we have
\[
J(\vbeta^{+},\mTheta^{+})
\leq
J(\widetilde{\vbeta},\widetilde{\mTheta})
-
c_l
\left\|
\begin{pmatrix}
\vbeta_l^{+}-\widetilde{\vbeta}_l\\
\operatorname{vec}(\mTheta_l^{+}-\widetilde{\mTheta}_l)
\end{pmatrix}
\right\|_2^2
\]
for some constant $c_l>0$, where $(\vbeta^{+},\mTheta^{+})$ differs from
$(\widetilde{\vbeta},\widetilde{\mTheta})$ only in block $l$. Also, $(\widetilde{\vbeta}_l, \widetilde{\mTheta}_l)$ denotes the current 
inner iterate for block $l$ before the update is accepted. Therefore,
\[
\|\vbeta_l^{(k+1)}-\vbeta_l^{(k)}\|_2
\longrightarrow 0,
\qquad
\|\mTheta_l^{(k+1)}-\mTheta_l^{(k)}\|_F
\longrightarrow 0,
\qquad l=1,\ldots,L.
\]

Now let $(\vbeta^*,\mTheta^*)$ be any limit point. Then there exists a
subsequence $\{(\vbeta^{(k_j)},\mTheta^{(k_j)})\}_{j\geq1}$ such that
\[
(\vbeta^{(k_j)},\mTheta^{(k_j)})
\longrightarrow
(\vbeta^*,\mTheta^*).
\]
For each block $l$, the update of the blockwise proximal-gradient satisfies the
first-order optimality condition
\[
\bm{0}
\in
\nabla_l f(\widetilde{\vbeta},\widetilde{\mTheta})
+
\frac{1}{t_l}
\Delta
+
\partial g_l(\vbeta_l^{+},\mTheta_l^{+}),
\]
where $\Delta = (\vbeta_l^+ - \widetilde{\vbeta}_l,
\operatorname{vec}(\mTheta_l^+ - \widetilde{\mTheta}_l))$, $g_l$ denotes the part of the SGPL penalty involving block
$(\vbeta_l,\mTheta_l)$, and $t_l$ is the step size (or learning rate). Along the convergent subsequence, the block increments
vanish, and by continuity of $\nabla_l f$ and closedness of the graph of the
convex sub-differential, passing to the limit gives
\[
\bm{0}
\in
\nabla_l f(\vbeta^*,\mTheta^*)+\partial g_l(\vbeta_l^*,\mTheta_l^*),
\qquad l=1,\ldots,L.
\]
Since the SGPL penalty is separable across the $X$-groups,
\[
g(\vbeta,\mTheta)=\sum_{l=1}^L g_l(\vbeta_l,\mTheta_l),
\]
the above blockwise inclusions imply
\[
\bm{0}
\in
\nabla f(\vbeta^*,\mTheta^*)+\partial g(\vbeta^*,\mTheta^*)
=
\partial J(\vbeta^*,\mTheta^*).
\]
Because $J$ is convex, the condition
$\bm{0}\in\partial J(\vbeta^*,\mTheta^*)$ is necessary and sufficient for
global optimality. Hence every limit point of the sequence generated by
Algorithm~\ref{alg:sgpl} is a minimiser of $J$.

\textbf{(iii) Backtracking termination.}
We verify termination of the backtracking line search at a generic inner
step. For fixed values of all blocks except block $l$, define the block
smooth loss as a function of
$
\vgamma_l =
(\vbeta_l^\top,\operatorname{vec}(\mTheta_l)^\top)^\top
$.
The gradient of the smooth loss is Lipschitz continuous in $\vgamma_l$ with
constant
\[
L_l = \frac{1}{n}\|\widetilde{\mX}_l\|_{\mathrm{op}}^2,
\]
where $\widetilde{\mX}_l$ is the augmented design matrix for group $l$,
and $\|\cdot\|_{\mathrm{op}}$ denotes the operator (spectral) norm,
i.e., the largest singular value of a matrix.

Hence, for any $t_l \leq 1/L_l$, the descent lemma gives
\[
f(\vgamma_l^+)
\leq
f(\widetilde{\vgamma}_l)
+
\langle \nabla_l f(\widetilde{\vgamma}_l),
\vgamma_l^+ - \widetilde{\vgamma}_l\rangle
+
\frac{1}{2t_l}
\|\vgamma_l^+ - \widetilde{\vgamma}_l\|_2^2 
\]
where $\langle . , . \rangle$ denotes the dot product. Therefore, for sufficiently small $t_l$, the quadratic upper bound used in
the proximal-gradient step is valid.

Since the backtracking procedure replaces $t_l$ by $\rho t_l$, with
$\rho\in(0,1)$, after finitely many reductions we must have
\[
t_l \leq \frac{1}{L_l}.
\]
Once $t_l \leq 1/L_l$, the quadratic upper bound holds and the proximal 
update produces a candidate $\vgamma_l^+$ satisfying the sufficient-decrease 
condition on the full SGPL objective (Steps~20--23 of 
Algorithm~\ref{alg:sgpl}), so the candidate is accepted. Hence the 
backtracking line search terminates in a finite number of steps. 

\end{proof}

% =====================================
% Proof Theorem 1
% ======================================
\subsection{Theorem 1}\label{supp:Theo1}

\begin{proof}
Let
\[
\vgamma^*
=
(\vbeta^{*\top},\operatorname{vec}(\mTheta^*)^\top)^\top,
\qquad
\widehat{\vgamma}
=
(\widehat{\vbeta}^{\top},\operatorname{vec}(\widehat{\mTheta})^\top)^\top,
\]
and define
\[
\vdelta=\widehat{\vgamma}-\vgamma^*.
\]
Then
\[
\widehat{\vy}-\vy^*
=
\widetilde{\mX}\vdelta.
\]

Since $(\widehat{\vbeta},\widehat{\mTheta})$ minimizes the SGPL objective,
\[
J(\widehat{\vbeta},\widehat{\mTheta})
\leq
J(\vbeta^*,\mTheta^*).
\]
Using $\vy = \vy^* + \vepsilon$, we expand
\[
\frac{1}{2n}\|\widehat{\vy} - \vy\|_2^2
=
\frac{1}{2n}\|\widehat{\vy} - \vy^* - \vepsilon\|_2^2
=
\frac{1}{2n}
\left(
\|\widehat{\vy} - \vy^*\|_2^2
- 2\langle \vepsilon, \widehat{\vy} - \vy^* \rangle
+ \|\vepsilon\|_2^2
\right).
\]

\noindent Substituting into the basic inequality
\[
\frac{1}{2n}\|\widehat{\vy} - \vy\|_2^2 + g(\widehat{\vgamma})
\le
\frac{1}{2n}\|\vy - \vy^*\|_2^2 + g(\vgamma^*),
\]
we obtain
\[
\frac{1}{2n}
\left(
\|\widehat{\vy} - \vy^*\|_2^2
- 2\langle \vepsilon, \widehat{\vy} - \vy^* \rangle
+ \|\vepsilon\|_2^2
\right)
+ g(\widehat{\vgamma})
\le
\frac{1}{2n}\|\vepsilon\|_2^2 + g(\vgamma^*),
\]
which can be rewritten
\[
\frac{1}{2n}
\|\widetilde{\mX}\vdelta\|_2^2
\leq
\frac{1}{n}
\langle \vepsilon,\widetilde{\mX}\vdelta\rangle
+
g(\vgamma^*)-g(\widehat{\vgamma}),
\]
where $g$ denotes the SGPL penalty.

\noindent We first control the stochastic term. For each column $\widetilde{\vx}_r$ of
$\widetilde{\mX}$,
\[
\widetilde{\vx}_r^\top \vepsilon
\sim
N(0,\sigma^2\|\widetilde{\vx}_r\|_2^2 \mI_n).
\]
By the column normalisation assumption
$\|\widetilde{\vx}_r\|_2^2/n\leq 1$, we have
\[
\frac{\widetilde{\vx}_r^\top \vepsilon}{n}
\sim
N\left(0,\frac{\sigma^2\|\widetilde{\vx}_r\|_2^2}{n^2} \mI_n\right),
\qquad
\frac{\sigma^2\|\widetilde{\vx}_r\|_2^2}{n^2}
\leq
\frac{\sigma^2}{n}.
\]
Hence, by a union bound over the $p(1+K)$ columns and the Gaussian tail bound,
for
\[
\lambda
=
A\sigma
\sqrt{\frac{\log\{p(1+K)\}}{n}},
\]
with $A>0$ sufficiently large,
\[
\left\|
\frac{1}{n}\widetilde{\mX}^{\top}\vepsilon
\right\|_{\infty}
\leq
\lambda
\]
with probability at least $1-2/\{p(1+K)\}$.

On this event, by applying the Holder's inequality,
\[
\frac{1}{n}
\langle \vepsilon,\widetilde{\mX}\vdelta\rangle
=
\left\langle
\frac{1}{n}\widetilde{\mX}^{\top}\vepsilon,
\vdelta
\right\rangle \leq \left\|
\frac{1}{n}\widetilde{\mX}^{\top}\vepsilon
\right\|_{\infty} \|\vdelta\|_1
\leq
\lambda\|\vdelta\|_1.
\]
where $\|\vdelta\|_1 = \|\vdelta_{\mathcal{H}}\|_1 + \|\vdelta_{\mathcal{H}^c}\|_1$. Since $\vdelta$ is partitioned into blocks
$\vdelta_j=(\delta_{\beta_j},\vdelta_{\vtheta_j})\in\mathbb{R}^{1+K}$, using the Cauchy-Schwarz inequality, we have
\[
\|\vdelta_{\mathcal{H}}\|_1
\leq
\sqrt{s_{\mathcal{H}}(1+K)}\,
\|\vdelta_{\mathcal{H}}\|_2.
\]

%----------------------------

\noindent For each inactive block $j\in\mathcal{H}^c$, by the Cauchy--Schwarz
inequality,
\[
\|\vdelta_j\|_1
\leq
\sqrt{1+K}\,\|\vdelta_j\|_2.
\]
Therefore,
\[
\|\vdelta_{\mathcal{H}^c}\|_1
=
\sum_{j\in\mathcal{H}^c}\|\vdelta_j\|_1
\leq
\sqrt{1+K}
\sum_{j\in\mathcal{H}^c}\|\vdelta_j\|_2
=
\sqrt{1+K}\,
\|\vdelta_{\mathcal{H}^c}\|_{2,1}.
\]

%%%%%%%%%%%%%%%%%%%%%%%%%%%

\noindent By the cone condition,
\[
\|\vdelta_{\mathcal{H}^c}\|_{2,1}
\le
\xi\,\|\vdelta_{\mathcal{H}}\|_{2,1}.
\]
Combining the two inequalities yields
\[
\|\vdelta_{\mathcal{H}^c}\|_1
\le
\sqrt{1+K}\,\xi\,\|\vdelta_{\mathcal{H}}\|_{2,1}.
\]

%------------------------------

\noindent Therefore,
\[
\|\vdelta\|_1
\leq
C_\xi
\sqrt{s_{\mathcal{H}}(1+K)}
\|\vdelta_{\mathcal{H}}\|_2,
\]
for a constant $C_\xi>0$ depending only on $\xi$.

Using the penalty comparison and the cone condition, the basic inequality can therefore be written as
\[
\frac{1}{2n}
\|\widetilde{\mX}\vdelta\|_2^2
\leq
C_1\lambda
\sqrt{s_{\mathcal{H}}(1+K)}
\|\vdelta_{\mathcal{H}}\|_2,
\]
where $C_1>0$ depends only on the cone constant and the penalty weights.

By Assumption~\ref{ass:gre},
\[
\frac{1}{n}
\|\widetilde{\mX}\vdelta\|_2^2
\geq
\phi^2\|\vdelta_{\mathcal{H}}\|_2^2.
\]
Hence
\[
\|\vdelta_{\mathcal{H}}\|_2
\leq
\frac{1}{\phi}
\frac{1}{\sqrt{n}}
\|\widetilde{\mX}\vdelta\|_2.
\]
Substituting this into the previous inequality gives
\[
\frac{1}{2n}
\|\widetilde{\mX}\vdelta\|_2^2
\leq
C_1\lambda
\sqrt{s_{\mathcal{H}}(1+K)}
\frac{1}{\phi}
\frac{1}{\sqrt{n}}
\|\widetilde{\mX}\vdelta\|_2.
\]
If $\widetilde{\mX}\vdelta=\bm{0}$, the prediction bound is immediate.
Otherwise, dividing both sides by
$\|\widetilde{\mX}\vdelta\|_2/\sqrt{n}$ gives
\[
\frac{1}{\sqrt{n}}
\|\widetilde{\mX}\vdelta\|_2
\leq
C_2
\frac{\lambda\sqrt{s_{\mathcal{H}}(1+K)}}{\phi}.
\]
Squaring both sides yields
\[
\frac{1}{n}
\|\widetilde{\mX}\vdelta\|_2^2
\leq
C_2^2
\frac{\lambda^2 s_{\mathcal{H}}(1+K)}{\phi^2}.
\]
Since
\[
\lambda^2
=
A^2\sigma^2
\frac{\log\{p(1+K)\}}{n},
\]
we obtain
\[
\frac{1}{n}
\|\widehat{\vy}-\vy^*\|_2^2
=
\frac{1}{n}
\|\widetilde{\mX}\vdelta\|_2^2
\leq
C
\frac{
\sigma^2 s_{\mathcal{H}}(1+K)\log\{p(1+K)\}
}{
n\phi^2
}.
\]

It remains to derive the estimation error bound. From the GRE condition,
\[
\|\vdelta_{\mathcal{H}}\|_2^2
\leq
\frac{1}{\phi^2}
\frac{1}{n}
\|\widetilde{\mX}\vdelta\|_2^2.
\]
Together with the cone condition,
\[
\|\vdelta_{\mathcal{H}^c}\|_2
\leq
\|\vdelta_{\mathcal{H}^c}\|_{2,1}
\leq
\xi\|\vdelta_{\mathcal{H}}\|_{2,1}
\leq
\xi\sqrt{s_{\mathcal{H}}}\|\vdelta_{\mathcal{H}}\|_2.
\]
Therefore,
\[
\|\vdelta\|_2^2
=
\|\vdelta_{\mathcal{H}}\|_2^2
+
\|\vdelta_{\mathcal{H}^c}\|_2^2
\leq
(1+\xi^2s_{\mathcal{H}})
\|\vdelta_{\mathcal{H}}\|_2^2.
\]
Absorbing the sparsity-dependent factor into the constant gives
\[
\|\vdelta\|_2^2
\leq
C'
\frac{
\sigma^2 s_{\mathcal{H}}(1+K)\log\{p(1+K)\}
}{
n\phi^4
}.
\]
Since
\[
\|\vdelta\|_2^2
=
\|\widehat{\vbeta}-\vbeta^*\|_2^2
+
\|\widehat{\mTheta}-\mTheta^*\|_F^2,
\]
the claimed estimation error bound follows.
\end{proof}

% ==================
% Section: Lambda max for Logistic model
% ===================

\section{Selection of
\texorpdfstring{$\lambda_{\max}$}{lambda max}
for logistic regression}
\label{sec:lambda_max_logistic}

For the logistic SGPL model, the regularisation path is constructed from a
sufficiently large value $\lambda_{\max}$ used as the upper endpoint of the
path. Because the intercept and the main effects of the modifying variables
are not penalised, we first fit the null logistic model containing only these
terms. Specifically, with $\vbeta=\bm{0}$ and $\mTheta=\bm{0}$, let
$$
\widehat{\eta}_{0i}
=
\widehat{\beta}_0
+
\vz_i^\top\widehat{\vtheta}_0,
\qquad
\widehat{\mu}_{0i}
=
\frac{\exp(\widehat{\eta}_{0i})}
     {1+\exp(\widehat{\eta}_{0i})},
$$
where $(\widehat{\beta}_0,\widehat{\vtheta}_0)$ are the maximum-likelihood
estimates from the unpenalised null model. Define the corresponding residual
vector as
$$
\vr_0
=
\vy-\widehat{\vmu}_0.
$$

For the $j$th main predictor, the negative gradient of the logistic negative
log-likelihood, evaluated at the null penalised solution, is
$$
\vg_j
=
\begin{pmatrix}
n^{-1}\vx_j^\top\vr_0\\[1mm]
n^{-1}\mZ^\top(\vx_j\odot\vr_0)
\end{pmatrix},
$$
where the first component corresponds to the main-effect coefficient
$\beta_j$, and the remaining components correspond to the interaction
coefficients $\vtheta_j$. This follows from the logistic gradients
$$
\nabla_{\beta_j}\ell
=
\frac{1}{n}\vx_j^\top(\vmu-\vy),
\qquad
\nabla_{\vtheta_j}\ell
=
\frac{1}{n}\mZ^\top
\left\{\vx_j\odot(\vmu-\vy)\right\}.
$$

\textbf{Remark on the evaluation point.}
The quantities $\widehat{\vmu}_0$, $\vr_0$, and hence
$\{\vg_j\}_{j=1}^p$ are evaluated once at the unpenalised null fit and are
held fixed when determining $\lambda_{\max}$. Thus, $\vg_j$ is not recomputed
for candidate values of $\lambda$; only the threshold $\alpha\lambda$ and
the right-hand side $(1-\alpha)\lambda$ vary during the one-dimensional
search.

For $0\leq\alpha<1$, under the parameterisation
$$
\lambda_1=\lambda(1-\alpha),
\qquad
\lambda_2=\lambda\alpha,
$$
a convenient conservative starting value can be obtained from the
predictor-level KKT condition. Let
$$
S_t(a)
=
\operatorname{sign}(a)(|a|-t)_+
$$
denote the element-wise soft-thresholding operator, with its vector extension
defined componentwise. A sufficient condition for the $j$th penalised
predictor block to remain zero is
$$
\left\|
S_{\alpha\lambda}(\vg_j)
\right\|_2
\leq
(1-\alpha)\lambda.
$$
We therefore choose $\lambda_{\max}$ as the smallest positive value satisfying
this condition simultaneously for all $j=1,\ldots,p$, namely
$$
\lambda_{\max}
=
\inf\left\{
\lambda>0:
\max_{1\leq j\leq p}
\left[
\left\|
S_{\alpha\lambda}(\vg_j)
\right\|_2
-
(1-\alpha)\lambda
\right]
\leq 0
\right\}.
$$

Since $\{\vg_j\}_{j=1}^p$ are fixed, the left-hand side is a deterministic
continuous function of $\lambda$ alone. It is monotone non-increasing in
$\lambda$: each component of $S_{\alpha\lambda}(\vg_j)$ is non-increasing in
absolute value as $\lambda$ increases, so
$\|S_{\alpha\lambda}(\vg_j)\|_2$ is non-increasing, while
$(1-\alpha)\lambda$ is non-decreasing. Consequently, $\lambda_{\max}$ can be
obtained efficiently by a one-dimensional root-finding procedure, such as
bisection, applied to the fixed vectors $\{\vg_j\}_{j=1}^p$.

For $\alpha=0$, no element-wise soft-thresholding is required and the
expression reduces to
$$
\lambda_{\max}
=
\max_{1\leq j\leq p}
\|\vg_j\|_2
=
\max_{1\leq j\leq p}
\sqrt{
\left(\frac{\vx_j^\top\vr_0}{n}\right)^2
+
\left\|
\frac{\mZ^\top(\vx_j\odot\vr_0)}{n}
\right\|_2^2
}.
$$

We use $\lambda_{\max}$ to define the upper end of the regularisation path
and construct a decreasing logarithmically spaced sequence of candidate
values of $\lambda$. Because the SGPL objective contains additional
overlapping group penalties, the expression above is used as a conservative
path-starting rule based on the predictor-level KKT condition, rather than
being interpreted as the exact smallest value of $\lambda$ associated with
the full overlapping SGPL penalty.

% ===================
% Section: Simulation results
% ===================

\section{Simulation results}

Before assessing performance across repeated simulations, we first examine the SGPL fit to a single simulated dataset generated under the design described in Section~\ref{sec:simulation}. The purpose of this analysis is illustrative: it provides a detailed view of how the proposed method behaves for one realization of the data, including the regularization path, selection of the tuning parameter by cross-validation, recovery of the true main and interaction coefficients, and the resulting sparsity and hierarchical structure. This complements the repeated-simulation study, which subsequently evaluates the stability and average performance of SGPL across independently generated datasets.

Figure~\ref{fig:results} presents the six diagnostic plots. The convergence plot (top left) confirms monotone decrease of the objective, and the CV curve (top middle) is well-behaved with $\widehat{\lambda}_{\min}$ and $\widehat{\lambda}_{1\text{se}}$ lying close together, indicating a stable, well-identified solution.

\begin{figure}[H]
    \centering
    \includegraphics[width=\linewidth]{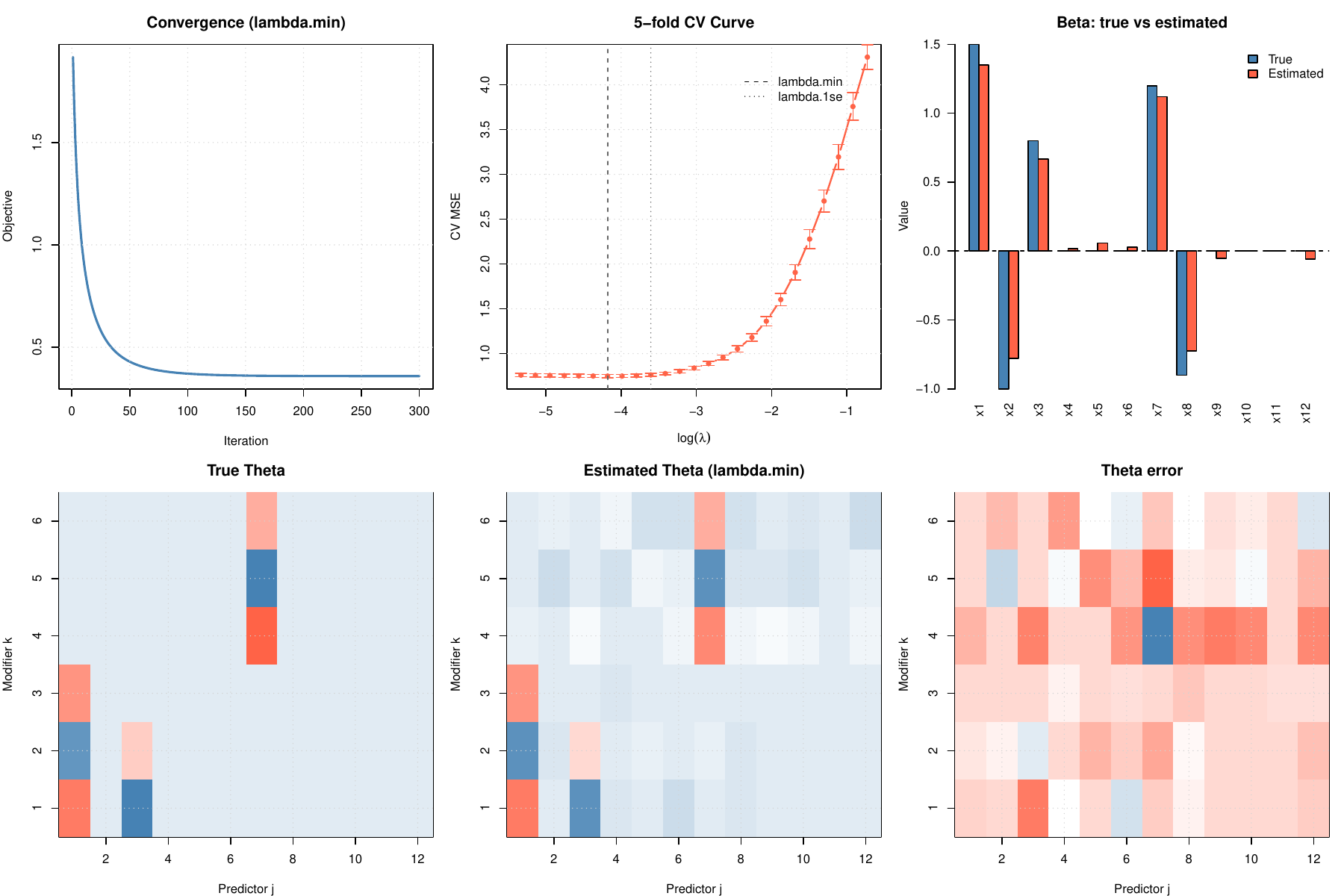}
    \caption{Simulation results for the SGPL at $\alpha = 0.5$ with $n=300$, $p=12$, $K=6$, $L=3$, $G=2$, SNR$=3$. \textit{Top left}: objective vs.\ iteration at $\widehat{\lambda}_{\min}$. \textit{Top middle}: five-fold CV curve; dashed and dotted lines mark $\widehat{\lambda}_{\min}$ and $\widehat{\lambda}_{1\text{se}}$. \textit{Top right}: true (blue) and estimated (red) $\vbeta$ at $\widehat{\lambda}_{\min}$. \textit{Bottom left}: true $\mTheta^*$. \textit{Bottom middle}: estimated $\widehat{\mTheta}$ at $\widehat{\lambda}_{\min}$. \textit{Bottom right}: error $\widehat{\mTheta} - \mTheta^*$ (blue = negative, white = zero, red = positive).}
    \label{fig:results}
\end{figure}

The bar chart (top right) shows perfect recall on the main effects: all five active predictors ($x_1, x_2, x_3, x_7, x_8$) are correctly identified and Group~3 ($x_9$--$x_{12}$) is fully suppressed. The interaction heatmaps (bottom row) confirm perfect recall on all eight nonzero entries of $\mTheta^*$, with the within-block zero $\theta^*_{3,3} = 0$ correctly recovered. Visual inspection of the error heatmap suggests that the largest estimation errors are concentrated in the Z-group 2 rows ($k=4$--$6$) of the $x_7$ and $x_8$ columns. This pattern is consistent with the effects of within-group correlation ($\rho = 0.4$), which can make it more difficult to distinguish correlated predictors under group penalization.

Table~\ref{tab:results} summarizes quantitative performance. At $\widehat{\lambda}_{\min}$, the test MSE of 0.562 is close to the irreducible noise floor $\sigma^2 \approx 0.481$, while the corresponding $R^2=0.879$ indicates strong predictive performance. The lower precision for $\vbeta$ (0.455) and $\mTheta$ (0.157) indicates the presence of false positive selections at $\widehat{\lambda}_{\min}$. 

\begin{table}[H]
\centering
\caption{Predictive performance, coefficient recovery, and support recovery on the held-out test set ($n_{\text{test}} = 500$). The irreducible noise variance is $\sigma^2 \approx 0.481$.}
\label{tab:results}
\footnotesize

%--------------------------------------------------
% Part (a)
%--------------------------------------------------

\begin{tabular}{lcccccccc}
\toprule
&
\multicolumn{2}{c}{Predictive}
&
\multicolumn{6}{c}{Support recovery}
\\
\cmidrule(lr){2-3}
\cmidrule(lr){4-9}
Solution
& Test MSE
& $R^2$
& $\beta$ Rec.
& $\beta$ Prec.
& $\beta$ F1
& $\Theta$ Rec.
& $\Theta$ Prec.
& $\Theta$ F1
\\
\midrule
$\widehat{\lambda}_{\min}$
& 0.562 & 0.879
& 1.000 & 0.455 & 0.625
& 1.000 & 0.157 & 0.271
\\

$\widehat{\lambda}_{1\text{se}}$
& 0.573 & 0.877
& 1.000 & 0.455 & 0.625
& 1.000 & 0.178 & 0.302
\\
\bottomrule
\end{tabular}

\vspace{0.4cm}

%--------------------------------------------------
% Part (b)
%--------------------------------------------------

\begin{tabular}{lccc}
\toprule
&
\multicolumn{2}{c}{Coefficient error}
\\
\cmidrule(lr){2-3}
Solution
& $\beta$ MSE
& $\Theta$ MSE
\\
\midrule
$\widehat{\lambda}_{\min}$
& 0.0113
& 0.00128
\\

$\widehat{\lambda}_{1\text{se}}$
& 0.0161
& 0.00141
\\
\bottomrule
\end{tabular}
\end{table}

\section{Corrections to the Public \texttt{svReg} Implementation}
\label{sec:supp_svreg_corrections}

The results for GPL reported in the main text (Table~\ref{tab:acc_comparison}) were obtained using the \texttt{svreg} R package, available from its developers' GitHub repository \citep{kim2021svreg}, but only after
several corrections were made to its reference implementation. Applying the package as
distributed produced either an outright numerical error or, at the package's default settings,
non-finite coefficient estimates on the ACC-CN500 dataset. This section documents each issue,
the correction applied, and the evidence supporting its validity. All corrected code is
provided in the accompanying software repository.

\subsection{Rank-Deficient Solve in Grouped Main-Predictor Updates}
\label{sec:supp_rank_deficiency}

Within the package's core fitting routine (\texttt{svReg1}), the update for a group of main
predictors $\ell$ with more than one member solves
\begin{equation}
\min_{\vbeta_{[\ell]}} \; \frac{1}{2N} \left\| r^{(-\ell)} - \mX_{[\ell]} \vbeta_{[\ell]} \right\|_2^2
+ \lambda_\ell \left\| \vbeta_{[\ell]} \right\|_2,
\qquad \lambda_\ell = \lambda(1-\alpha)\sqrt{p_\ell},
\label{eq:supp_grouplasso_subproblem}
\end{equation}
a standard weighted group-lasso subproblem, where $r^{(-\ell)}$ is the partial residual with
group $\ell$'s fit removed. The reference implementation initializes this subproblem with the
closed-form ordinary least squares estimate
\begin{equation}
\widehat\vbeta_{[\ell]}^{(0)} = \left( \mX_{[\ell]}^\top \mX_{[\ell]} \right)^{-1} \mX_{[\ell]}^\top r^{(-\ell)},
\label{eq:supp_ols_init}
\end{equation}
before refining it via coordinate-wise golden-section search. Equation~\eqref{eq:supp_ols_init}
requires $\mX_{[\ell]}^\top \mX_{[\ell]}$ to be invertible, i.e., $\mathrm{rank}(\mX_{[\ell]}) = p_\ell$.
This condition fails whenever a group's size exceeds the training sample size
($p_\ell > N_{\text{train}}$) or the group contains exactly collinear columns. In the
ACC-CN500 application, the largest chromosome-based group ($p_\ell = 350$) exceeds every
cross-validation fold's training sample size ($N_{\text{train}} = 60$), so
$\mX_{[\ell]}^\top \mX_{[\ell]}$ is guaranteed rank-deficient and \eqref{eq:supp_ols_init} fails with
a matrix-inversion error.

\paragraph{Correction.} We replaced the initialization-and-refinement step with an accelerated
proximal gradient (FISTA) solver \citep{beck2009fast} targeting the identical convex objective
\eqref{eq:supp_grouplasso_subproblem}. Writing $f(\vbeta) = \frac{1}{2N}\|r^{(-\ell)} -
\mX_{[\ell]}\vbeta\|_2^2$, each iteration computes
\begin{equation}
u = z_k - \frac{1}{L}\nabla f(z_k), \qquad
\vbeta_{k+1} = \left(1 - \frac{\lambda_\ell / L}{\|u\|_2}\right)_{+} u,
\end{equation}
with Nesterov momentum and adaptive restart \citep{o2015adaptive}, and $L =
\sigma_{\max}(\mX_{[\ell]})^2 / N$ computed via the exact largest singular value (inexpensive here
since $N \le 75$). This update requires no matrix inversion and is well-defined regardless of
whether $\mX_{[\ell]}^\top \mX_{[\ell]}$ is singular.

\paragraph{Validation.} Because the reference implementation cannot run at all when $p_\ell >
N$, direct comparison is impossible in that regime; we instead verified first-order optimality
directly via the subgradient stationarity residual of \eqref{eq:supp_grouplasso_subproblem} at
the FISTA solution (residual $< 10^{-7}$ in all tested cases with $p_\ell > N$). For
well-posed cases ($p_\ell < N$, varying correlation structure and group size), we confirmed the
FISTA solution agrees with the reference (unpatched) implementation's solution across the full
parameter path: maximum absolute differences in $\widehat\vbeta$, $\widehat\mTheta$, $\widehat\vbeta_0$,
$\widehat\vtheta_0$, the objective value, and fitted values were all below $10^{-3}$ across five
independent random settings (Table~\ref{tab:supp_validation}).

\begin{table}[ht]
\centering
\caption{Agreement between the FISTA-corrected solver and the reference implementation across
five well-posed validation settings. Small nonzero differences reflect differing convergence
tolerances between the two numerical procedures rather than disagreement in the underlying
solution.}
\label{tab:supp_validation}
\begin{tabular}{cccccc}
\toprule
Setting & $N$ & $p_\ell$ & Correlation & Max.\ $|\Delta\hat\beta|$ & Max.\ $|\Delta\hat\Theta|$ \\
\midrule
1 & 60 & 15 & Low  & $1.45\times10^{-4}$ & $2.75\times10^{-11}$ \\
2 & 60 & 25 & High & $5.70\times10^{-8}$ & $1.99\times10^{-10}$ \\
3 & 80 & 10 & Low  & $4.18\times10^{-4}$ & $2.32\times10^{-10}$ \\
4 & 80 & 30 & High & $1.29\times10^{-4}$ & $4.39\times10^{-11}$ \\
5 & 50 & 20 & High & $0$ & $0$ \\
\bottomrule
\end{tabular}
\end{table}

\subsection{Cross-Validation Standardization Correction}
\label{sec:supp_cv_standardization}

The package's cross-validation routine (\texttt{cv.svReg}) standardizes each held-out test fold
using that fold's own mean and standard deviation prior to scoring, rather than the moments of
the corresponding training fold on which the model was fit. Because coefficient estimates are
obtained on the training-standardized scale, applying them to independently-rescaled test data
introduces a scale mismatch inconsistent with standard out-of-sample evaluation practice.

\noindent We modified the cross-validation routine to compute standardization
moments from the training fold only, and apply those same moments -- rather than the test
fold's own moments -- when transforming held-out data prior to scoring. This mirrors standard
practice for held-out evaluation and is consistent with the approach used for PL and SGPL
elsewhere in this manuscript.

\subsection{Numerical Instability in the Joint
\texorpdfstring{$(\beta,\Theta)$}{(beta, Theta)}
Update}
\label{sec:supp_joint_update}

Groups for which both the main-effect and modifying-effect coefficients are active are updated
via a majorization--minimization step (Algorithm 1, Step 2c of \citealp{kim2021svreg}) using a
\emph{fixed} step size $t = \texttt{tt}$, with no backtracking line search and, in the
reference implementation, no bound on the number of times this update may be repeated within a
single call. On the ACC-CN500 dataset, at the package's documented default (\texttt{tt = 0.1}),
this update diverged: coefficient magnitudes for the largest chromosome-based group grew
across successive outer iterations to order $10^{149}$--$10^{153}$ before triggering a
non-finite-value error. Reducing the step size to \texttt{tt = 0.01} delayed but did not
prevent divergence, which instead surfaced in a small, unrelated group (size 2) after 44 outer
iterations, indicating the instability originates in the fixed-step update mechanism itself
rather than being specific to any one group's size or conditioning.

\paragraph{Correction.} We made two changes. First, we added a hard iteration cap (100
passes) to the previously unbounded inner loop governing this update, with a warning emitted
if the cap is reached without the screening decision stabilizing. Second, we determined
empirically that \texttt{tt = 0.001} -- two orders of magnitude below the package default --
produces finite, stable coefficient trajectories across the full default outer-iteration budget
($\texttt{iter} = 500$) on this dataset. We adopted this value for all GPL results reported in
the main text. We view this as a practical correction rather than a complete solution: a fixed
step size, however small, is not guaranteed to satisfy the majorization condition for every
possible group structure, and a proper backtracking line search -- selecting the step size
adaptively at each iteration so that the majorized objective is guaranteed to decrease -- would
be a more principled long-term solution. We report the fixed-step correction here because it
is sufficient to reproduce stable results on the present dataset; the corrected code includes
diagnostic instrumentation (reporting the specific intermediate quantities in the
majorization--minimization update, and the outer iteration and group at which any future
instability occurs) to facilitate further investigation.

\subsection{Computational Cost}
\label{sec:supp_computational_cost}

With these corrections applied, fitting GPL via five-fold cross-validation over a 25-value
$\lambda$ grid on the full ACC-CN500 dataset required approximately 11.3 hours on a standard
desktop workstation. We do not interpret this runtime as evidence about the relative
computational efficiency of the GPL method itself: the \texttt{svreg} package is implemented
in base R, whereas our SGPL implementation makes substantial use of Rcpp, so a direct runtime
comparison would conflate implementation language with algorithmic cost. We report this figure
only for transparency and to assist other researchers working with this implementation on
similarly structured data (a single main-predictor group of $p_\ell = 350$ with $N = 75$).

\subsection{Reproducibility}
\label{sec:supp_reproducibility}

All corrected source files (\texttt{basic\_functions\_patched.R}, \texttt{svreg\_patched.R},
\texttt{cv\_svReg\_corrected.R}) and the validation scripts used to generate
Table~\ref{tab:supp_validation} are included in the accompanying repository at
\url{https://github.com/edelweiss611428/SGPL-code} The original, unmodified \texttt{svreg} package is retained alongside the
corrected files for direct comparison.

\section{Efficient SGPL Computation}

Naive computation of the SGPL objective function during coordinate descent is computationally inefficient. In each blockwise step, parameters for only one predictor group are updated while the non-target parameters remain fixed. Recomputing the full fitted vector $\widehat{\vy}$ or evaluating the global penalty from scratch at every inner iteration introduces significant redundant operations.

To eliminate this redundancy, we exploit the additive structure of both the linear model and the regularisation penalty. While the derivation is developed for the SGPL linear regression model, the same logic can be extended to other methods, such as SGPL logistic regression, which, although nonlinear, still involve the computation of a linear predictor. The proposed implementation is available in the R package \textsf{SGPL}.

\subsection{Efficient Blockwise Residual Computation}

Recall the vector of predicted values $\widehat{\vy} \in \mathbb{R}^n$ defined by the pliable interaction model:
\begin{align}\label{linear_model_PL_suppl}
    \widehat{\vy} = \widehat\beta_0 \vone_n + \mZ \widehat\vtheta_0 + \mX\widehat\vbeta + \sum_{j=1}^p (\vx_j \odot \mZ)\widehat\vtheta_{\bullet j}.
\end{align}

\noindent Partitioning the $p$ predictors into $L$ non-overlapping index groups $\{\mathcal{G}_l^x\}_{l=1}^L$, we can decompose $\widehat{\vy}$ into the contribution from the $l$-th active predictor group and the cumulative contribution from all non-$l$ groups ($\mathcal{G}_{-l}^x$):
\begin{align}\label{eq:yhat_split}
    \widehat{\vy} 
    &= \underbrace{\widehat\beta_0 \vone_n + \mZ\widehat\vtheta_0 + \sum_{k \neq l} \left( \mX_k \odot (\vone_n \widehat \vbeta_{k}^\top + \mZ\widehat\vtheta_{k\bullet }) \right) \vone_{p_k}}_{\widehat{\vy}^{(-l)} \text{ (non-group } l\text{ contribution)}} 
    + \underbrace{\left( \mX_l \odot (\vone_n \widehat\vbeta_{l}^\top + \mZ\widehat\vtheta_{l\bullet }) \right) \vone_{p_l}}_{\widehat{\vy}^{(l)} \text{ (group } l\text{ contribution)}},
\end{align}
where $\mathbf{\widehat\vtheta}_{l\bullet} \in \mathbb{R}^{K \times p_l}$ denotes the subset of interaction parameters corresponding to predictor group $l$. By maintaining a cache matrix $\mathbf{\Gamma} = \mZ\widehat{\bm\Theta} \in \mathbb{R}^{n \times p}$, the local contribution of block $l$ reduces to:
\begin{equation}
\widehat{\vy}^{(l)} = \left[ \mX_l \odot \left( \vone_n \widehat\vbeta_{l}^\top + \mathbf{\Gamma}_{l\bullet} \right) \right] \vone_{p_l}.
\end{equation}
During the update step for group $l$, $\widehat{\vy}^{(-l)}$ is strictly invariant. Consequently, we avoid recomputing the full prediction vector by maintaining the global residual $\vr = \vy - \widehat{\vy}$. When updating block $l$:

\begin{itemize}
\item The partial residual $\vr^{(-l)}$ is computed on the fly by adding back the current contribution of group $l$ to the residual:
\begin{align}\label{eq:partial_residual}
\vr^{(-l)} = \vy - \widehat{\vy}^{(-l)} = \vr + \widehat{\vy}^{(l)}.
\end{align}
\item Given the partial residual $\vr^{(-l)}$, the candidate residual $\vr_l^+$ is evaluated using the updated parameters for group $l$:
\begin{align}\label{eq:residual_update}
\vr_l^+ = \vr^{(-l)} - \widehat{\vy}^{(l),+}.
\end{align}
\end{itemize}

These operations are translated into the following two procedures.

\begin{algorithm}[H]
\caption{Efficient Blockwise Residual Computation for SGPL}
\small
\label{alg:sgpl_residual}

\begin{algorithmic}[1]

\Require
Current residual $\vr$;
predictor block matrix $\mX_l$;
modifier matrix $\mZ$;
current block coefficients $\widehat{\bm{\beta}}_{l}$;
cached modifier contribution $\mathbf{\Gamma}_{l\bullet}$;
candidate block coefficients $\widehat{\bm{\beta}}_{l}^{+}$;
candidate modifier coefficients $\widehat{\vtheta}_{l\bullet}^{+}$.

\Ensure
Partial residual $\vr^{(-l)}$ and updated block residual $\vr_l^{+}$.

\Procedure{ComputePartialResidual}{$\vr, \mX_l, \widehat{\bm{\beta}}_{l}, \mathbf{\Gamma}_{l\bullet}$}
    \State $\widehat{\vy}^{(l),\mathrm{old}} \leftarrow 
    \left[
    \mX_l \odot
    \left(
    \vone_n\widehat{\bm{\beta}}_{l}^{\top}
    +
    \mathbf{\Gamma}_{l\bullet}
    \right)
    \right]
    \vone_{p_l}$
    \State \Return $\vr^{(-l)} \leftarrow \vr+\widehat{\vy}^{(l),\mathrm{old}}$
\EndProcedure

\Procedure{ComputeBlockResidual}{$\vr^{(-l)}, \mX_l, \mZ,
\widehat{\bm{\beta}}_{l}^{+},
\widehat{\vtheta}_{l\bullet}^{+}$}
    \State $\widehat{\vy}^{(l),+} \leftarrow
    \left[
    \mX_l \odot
    \left(
    \vone_n(\widehat{\bm{\beta}}_{l}^{+})^{\top}
    +
    \mZ\widehat{\vtheta}_{l\bullet}^{+}
    \right)
    \right]
    \vone_{p_l}$
    \State \Return $\vr_l^+ \leftarrow \vr^{(-l)}-\widehat{\vy}^{(l),+}$
\EndProcedure
\end{algorithmic}
\end{algorithm}

\subsection{Efficient Blockwise Penalty Computation}

The exact same blockwise separation applies to the objective function's penalty term. The full objective function $J(\widehat\vbeta, \widehat{\mathbf{\Theta}})$ is:
\begin{equation}
\begin{aligned}
J(\widehat\vbeta, \widehat{\mathbf{\Theta}})
&=
\frac{1}{2n} \|\vy - \widehat{\vy}\|_2^2 + \lambda(1-\alpha)
\left(
    \sum_{j=1}^{p} \|(\widehat\beta_{j}, \widehat\vtheta_{j })\|_2
    + \sum_{l=1}^{L} \sqrt{p_l} \,
      \|(\widehat\vbeta_l, \widehat\vtheta_{l\bullet })\|_2
\right. \\
&\qquad \left.
    + \sum_{l=1}^{L} \sum_{g=1}^{G}
      \frac{\sqrt{p_g}}{\sqrt{1+K}}
      \|\widehat\vtheta_{lg}\|_2
\right)  + \lambda\alpha
\left(
    \|\widehat\vbeta\|_1
    + \sum_{j=1}^{p} \|\widehat\vtheta_{j}\|_1
\right).
\end{aligned}
\label{eq:SGPL_suppl}
\end{equation}

Because every regularization term is separable across predictor groups $l = 1, \dots, L$, the total penalty $P_{\text{total}}(\widehat\vbeta, \widehat{\mathbf{\Theta}})$ can be factored as:
\[
P_{\text{total}}(\widehat\vbeta, \widehat{\mathbf{\Theta}}) = P_{-l} + P_l\left( \widehat\vbeta_{l}, \widehat{\vtheta}_{l\bullet} \right),
\quad \text{where } P_{-l} = \sum_{k \neq l} P_k\left( \widehat\vbeta_{k}, \widehat{\vtheta}_{k\bullet} \right).
\]
Since $P_{-l}$ remains constant throughout all inner updates to block $l$, line-search checks and cost evaluations do not need to re-evaluate penalty norms across the un-updated $L-1$ groups. Instead, the updated penalty is evaluated via an $\mathcal{O}(1)$ block-delta update:
\[
P_{\text{total}}^+ = P_{\text{total}} - P_l^{\text{old}} + P_l^+.
\]

These operations are translated into the following two procedures.

\begin{algorithm}
\small
\caption{Efficient Blockwise SGPL Penalty Computation}
\small
\label{alg:sgpl_penalty}
\begin{algorithmic}[1]

\Require
Current block parameters $\widehat{\bm{\beta}}_{l}$ and $\widehat\vtheta_{l\bullet}$;
previous block penalty $P_l^{\mathrm{old}}$;
total penalty $P_{\mathrm{total}}$;
penalty parameters $\lambda_1,\lambda_2$;  modifier groups $\{\mathcal{G}_g^z\}_{g=1}^G$; total modifier count $K$

\Ensure
Updated block penalty $P_l^+$ and total penalty $P_{\mathrm{total}}^+$.

\Procedure{ComputeBlockPenalty}{$\widehat{\bm{\beta}}_{l},\widehat\vtheta_{l\bullet},\lambda_1,\lambda_2,\{\mathcal{G}_g^z\}_{g=1}^G,K$}
    \State $p_{\mathrm{pred}} \leftarrow \lambda_1 \sum_{j\in\mathcal{G}_l^x}
    \sqrt{\widehat\beta_j^2+\|\widehat\vtheta_{j\bullet}\|_2^2}$
    \State $p_{\mathrm{joint}} \leftarrow \lambda_1\sqrt{p_l}
    \sqrt{\|\widehat{\bm{\beta}}_{l}\|_2^2+
    \|\widehat\vtheta_{l\bullet}\|_F^2}$
    \State $p_{\mathrm{mod}} \leftarrow \lambda_1
    \sum_{g=1}^G
    \frac{\sqrt{p_g}}{\sqrt{1+K}} \|\widehat\vtheta_{lg}\|_F$
    \State $p_{L_1} \leftarrow \lambda_2
    \left(
    \|\widehat{\bm{\beta}}_{l}\|_1+
    \|\widehat\vtheta_{l\bullet}\|_1
    \right)$
    \State \Return $P_l \leftarrow
    p_{\mathrm{pred}}+p_{\mathrm{joint}}+p_{\mathrm{mod}}+p_{L_1}$
\EndProcedure

\Procedure{UpdateTotalPenalty}{$P_{\mathrm{total}},P_l^{\mathrm{old}},P_l^+$}
    \State \Return $P_{\mathrm{total}}^+
    \leftarrow P_{\mathrm{total}}-P_l^{\mathrm{old}}+P_l^+$
\EndProcedure

\end{algorithmic}
\end{algorithm}

\subsection{Efficient Blockwise Coordinate-Descent Algorithm for SGPL}

The resulting blockwise coordinate-descent procedure is summarised in Algorithm~\ref{alg:sgpl_cached}.

\begin{algorithm}
\small
\setstretch{0.85}
\caption{Blockwise Coordinate-Descent Algorithm with Efficient Caching for SGPL}
\label{alg:sgpl_cached}
\begin{algorithmic}

\Require
Data $(\mathbf{X},\mathbf{Z},\mathbf{y})$;
penalty parameters $\lambda,\alpha$;
groups $\{\mathcal{G}_l^x\}_{l=1}^L$, $\{\mathcal{G}_g^z\}_{g=1}^G$;
step size $t$, contraction factor $\rho$;
tolerances $\epsilon_{\rm out},\epsilon_{\rm in}$;
initial parameters $(\beta_0^{(0)},\bm{\theta}_0^{(0)},\bm{\beta}^{(0)},\mathbf{\Theta}^{(0)})$.

\Ensure
$\widehat{\beta}_0,\widehat{\bm{\theta}}_0,\widehat{\bm{\beta}},\widehat{\mathbf{\Theta}}$.

\State Set $\lambda_1=\lambda(1-\alpha)$ and $\lambda_2=\lambda\alpha$.
\State Initialise $\mathbf{\Gamma}\leftarrow\mathbf{Z}\mathbf{\Theta}^{(0)}$.
\State $\mathbf{r}\leftarrow
\textproc{ComputeResidual}
(\mathbf{y},\beta_0^{(0)},\bm{\theta}_0^{(0)},\mathbf{X},
\bm{\beta}^{(0)},\mathbf{\Gamma})$.
\State $P_{\mathrm{total}}\leftarrow
\textproc{ComputeTotalPenalty}
(\bm{\beta}^{(0)},\mathbf{\Theta}^{(0)},\lambda_1,\lambda_2,
\{\mathcal{G}_l^x\},\{\mathcal{G}_g^z\},K)$.

\Repeat
\State Update $(\beta_0^{(k)},\bm{\theta}_0^{(k)})$ by fitting the OLS model $\mathbf{r}\sim(\mathbf{1}_n,\mathbf{Z})$.
\State Update global residual \textbf{r} for intercept change.

\For{$l=1,\ldots,L$}

\State $\mathbf{r}^{(-l)}
\leftarrow
\textproc{ComputePartialResidual}
(\mathbf{r},\mathbf{X}_l,
\bm{\beta}_{l}^{(k)},
\mathbf{\Gamma}_{l\bullet})$.

\State $P_l^{\mathrm{old}}\leftarrow
\textproc{ComputeBlockPenalty}
(\bm{\beta}_{l}^{(k)},
\mathbf{\Theta}_{l\bullet}^{(k)},
\lambda_1,\lambda_2,\{\mathcal{G}_g^z\},K)$.

\State Initialize
$\widetilde{\bm{\beta}}_{l}\leftarrow\bm{\beta}_{l}^{(k)}$,
$\widetilde{\bm{\theta}}_{l\bullet }\leftarrow
\vtheta_{l\bullet}^{(k)}$,
$t_l\leftarrow t$.
\State Check KKT conditions and skip inner loop if satisfied.
\Repeat

\State Apply gradient step and proximal updates:
\[
(\bm{\beta}_{l}^+,\bm{\theta}_{l\bullet}^+)
\leftarrow
\textproc{NestedProximalUpdate}
(\mathbf{r},\widetilde{\bm{\beta}}_{l},
\widetilde{\bm{\theta}}_{l\bullet},t_l).
\]

\State Compute candidate's residual:
\[
\mathbf{r}_l^+\leftarrow
\textproc{ComputeBlockResidual}
(\mathbf{r}^{(-l)},\mathbf{X}_l, \mathbf{Z},
\bm{\beta}_{l}^+,\bm{\theta}_{l\bullet}^+).
\]

\State Compute candidate's penalty:
\[
P_l^+\leftarrow
\textproc{ComputeBlockPenalty}
(\bm{\beta}_{l}^+,\bm{\theta}_{l\bullet}^+,
\lambda_1,\lambda_2,\{\mathcal{G}_g^z\},K).
\]

\State Compute candidate's total penalty:
\[
P_{\mathrm{total}}^+
\leftarrow
\textproc{UpdateTotalPenalty}
(P_{\mathrm{total}},P_l^{\mathrm{old}},P_l^+).
\]

\If{
$\textproc{ComputeCost}(\mathbf{r}_l^+,P_{\mathrm{total}}^+)
>
\textproc{ComputeCost}(\mathbf{r},P_{\mathrm{total}})$}
\State $t_l\leftarrow\rho t_l$
\Else
\State Accept update:
$\widetilde{\bm{\beta}}_{l}\leftarrow\bm{\beta}_{l}^+$,
$\widetilde{\bm{\theta}}_{l\bullet}\leftarrow\bm{\theta}_{l\bullet}^+$,
$P_{\mathrm{total}}\leftarrow P_{\mathrm{total}}^+$.
\State Update residual and total penalty: $(\mathbf{r}, P_{\text{total}}) \leftarrow (\mathbf{r}^+, P_{\text{total}}^+)$.
\EndIf

\Until{inner convergence}
\State Update block parameters:  $(\bm{\beta}_{l}, \bm{\theta}_{l\bullet}) \leftarrow (\widetilde{\bm{\beta}}_{l}, \widetilde{\bm{\theta}}_{l\bullet})$.
\State Update cache matrix:
\[
\mathbf{\Gamma}_{l\bullet}
\leftarrow
\mathbf{\Gamma}_{l\bullet}
+
\mathbf{Z}
\left(
\widetilde{\bm{\theta}}_{l\bullet}
-
\vtheta_{l\bullet}^{(k)}
\right).
\]
\State Update $(\beta_0^{(k)},\bm{\theta}_0^{(k)})$ by fitting the OLS model $\mathbf{r}\sim(\mathbf{1}_n,\mathbf{Z})$.

\EndFor

\Until{outer convergence}

\State \Return
$(\widehat{\beta}_0,\widehat{\bm{\theta}}_0,
\widehat{\bm{\beta}},\widehat{\mathbf{\Theta}})$.

\end{algorithmic}
\end{algorithm}

\subsection{Time Complexity Analysis}

Let $p_l = |\mathcal{G}_l^x|$ denote the size of predictor group $l$, 
$p=\sum_{l=1}^L p_l$the total number of predictors, $n$ the sample size, 
and $K$ the number of modifier variables. Assume that the $L$ predictor 
groups contain an equal number of predictors, so that
\[
p_l = \frac{p}{L}, \qquad l=1,\ldots,L.
\]

\begin{itemize}

\item \textbf{Naive Computation:} Consider a block update for predictor group $l$. Without caching, the fitted-value contribution associated with the modifier effects must be recomputed using all
$p$ predictors. For each of the $p$ predictors, the computation involves the $n$ observations and $K$ modifier variables, resulting in a cost of $
\mathcal{O}(npK)$ per block update. Since there are $L$ predictor groups, a complete outer iteration requires
$L$ such block updates. Hence,
\[
\mathcal{T}_{\mathrm{naive}}
=
\mathcal{O}(LnpK).
\]

Using \(L=p/p_l\), this can equivalently be written as
\[
\mathcal{T}_{\mathrm{naive}}
=
\mathcal{O}\left(
\frac{p}{p_l}npK
\right)
=
\mathcal{O}\left(
\frac{np^2K}{p_l}
\right).
\]

\item \textbf{Efficient Blockwise Updating:} With the proposed caching strategy, contributions from predictor groups that are not being updated are stored and do not need to be recomputed. During the
update of group \(l\), only the \(p_l\) predictors belonging to the active group
need to be processed. Consequently, the modifier-related computation for a single block update
requires $\mathcal{O}(np_lK)$
operations. Summing over all \(L\) predictor groups, the computational cost of one complete
outer iteration is
\[
\mathcal{T}_{\mathrm{efficient}}
=
\mathcal{O}\left(
\sum_{l=1}^L np_lK
\right) = \mathcal{O}(npK).
\]

\end{itemize}

The resulting computational speedup is therefore
\[
\frac{\mathcal{T}_{\mathrm{naive}}}
     {\mathcal{T}_{\mathrm{efficient}}}
=
\mathcal{O}\Big(\frac{p}{p_l}\Big).
\]

This result also covers the extreme case in which every predictor forms its
own group. In this setting, $p_k = 1$ and $L = p$ so that
\[
\mathcal{T}_{\mathrm{naive}}
=
\mathcal{O}(np^2K),
\qquad
\mathcal{T}_{\mathrm{efficient}}
=
\mathcal{O}(npK),
\]
and the resulting speedup is
\[
\frac{\mathcal{T}_{\mathrm{naive}}}
     {\mathcal{T}_{\mathrm{efficient}}}
=
\mathcal{O}(p).
\]

Hence, the computational advantage increases as the predictor groups become
smaller, reaching an \(\mathcal{O}(p)\) speedup when each predictor constitutes
an individual group. When sparsity is present, the computational gain can be further increased, as inactive groups can be excluded from the updates, leading to an additional reduction proportional to the inverse sparsity level.

\subsection{Runtime Evaluation}

To evaluate the computational efficiency gained through partial-residual tracking and global metric caching, we compared the optimised implementation, \texttt{fast\_sgpl\_fit\_cpp6}, against the baseline implementation without caching, \texttt{sgpl\_fit\_cpp}, both available in the \texttt{SGPL} R package.

\paragraph{Data Generating Process (DGP)} For all simulation scenarios, data were generated under a linear regression model:
\[
\vy = \mX \vbeta + \sum_{k=1}^K (\mX \odot \mZ_{\cdot, k}) \mTheta_{k, \cdot} + \varepsilon, \quad \varepsilon \sim \mathcal{N}(0, \sigma^2 I_n)
\]
where $\mX \in \mathbb{R}^{n \times p}$ and $\mZ \in \mathbb{R}^{n \times K}$ are feature matrices whose rows are sampled independently from a multivariate normal distribution $\mathcal{N}(0, \mSigma)$ with an autoregressive $\text{AR}(1)$ correlation structure ($\rho = 0.5$). Predictor features $\mX$ and modifier features $\mZ$ were partitioned into $L$ and $G$ equal contiguous groups, respectively. True coefficient vector $\vbeta$ and interaction matrix $\mTheta$ were generated with block-wise sparsity to reflect realistic sparse group structures.

\paragraph{Simulation Scenarios}
We benchmarked execution times across three distinct structural regimes:
\begin{itemize}
    \item \textbf{Scenario 1 (Sample Size Scaling):} Fixed $p = K = 40$, $L = G = 20$, while varying the sample size $n \in \{100, 300, 1000, 3000, 10000\}$.
    \item \textbf{Scenario 2 (Predictor Dimension Scaling):} Fixed $n = 100$ and $K = 20$, $G = 10$, while varying predictor dimensions $p \in \{10, 20, 40,80,160\}$ and $L = p/2$.
    \item \textbf{Scenario 3 (Modifier Dimension Scaling):} Fixed $n = 100$ and $p = 20$, $L = 10$, while varying modifier dimensions $K \in \{10, 20, 40,80\}$ and $G = K/2$.
\end{itemize}

\paragraph{Benchmarking Metrics}
Execution speed was measured in R using elapsed wall-clock time via the \texttt{system.time()} function. For each parameter configuration, both algorithms were fitted over 50 independent replications using identical random seeds, initial values, hyperparameters, and convergence tolerances. Reported runtimes are summarised by the median across replications.

\paragraph{Results}
Figure~\ref{fig:runtime_eval} compares the runtime of the proposed cached implementation against the baseline implementation without caching. The cached algorithm consistently outperforms the baseline across all experimental settings. The improvement is most pronounced as the number of predictors ($p$) increases, where the runtime grows much more slowly because the cached interaction matrix avoids repeatedly recomputing contributions that remain unchanged between coordinate updates. Similar speedups are observed when increasing the number of modifying variables ($K$) and the sample size ($n$), although the relative gain is less dramatic since both implementations must still perform operations whose cost scales with these quantities. Overall, the results demonstrate that the proposed caching strategy substantially improves computational efficiency while preserving the same optimisation procedure, with the greatest benefit achieved in large-$p$ settings.

\begin{figure}[t]
    \centering
    \includegraphics[width=\textwidth]{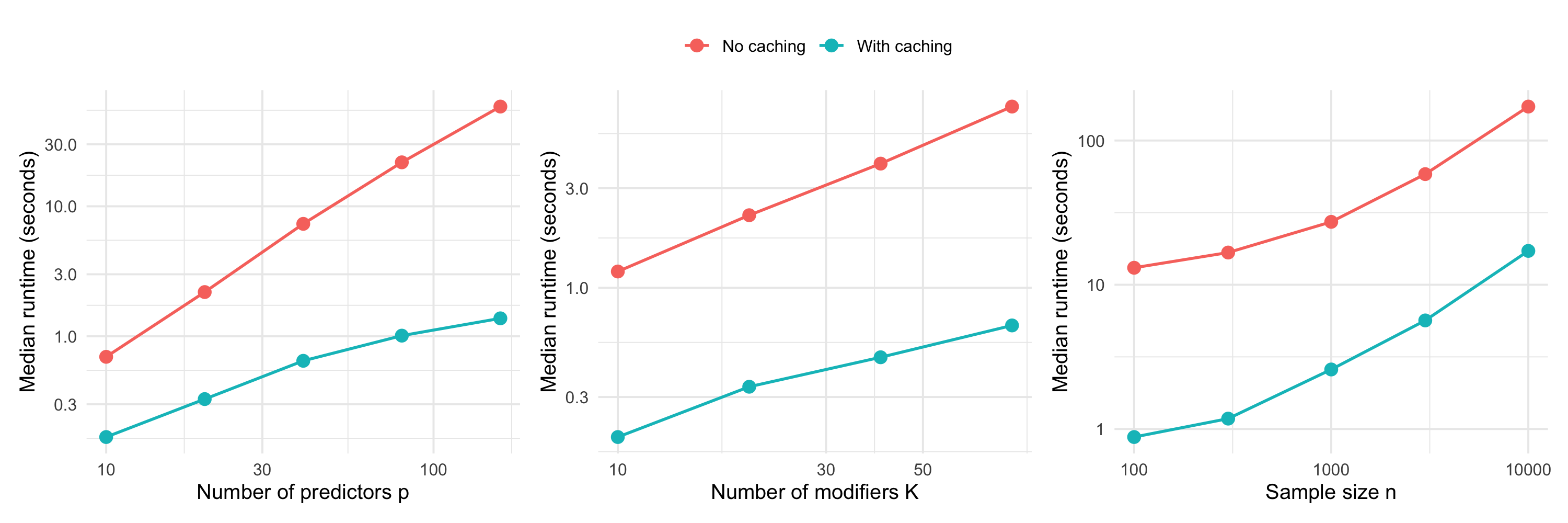}
    \caption{Median runtime comparison between the blockwise coordinate-descent algorithm with and without matrix caching. The three panels investigate scalability with respect to the number of predictors ($p$), the number of modifying variables ($K$), and the sample size ($n$), respectively. Both axes are displayed on logarithmic scales.}
    \label{fig:runtime_eval}
\end{figure}

\end{document}